\documentclass[submission]{eptcs} 
\providecommand{\event}{LINEARITY 2014} 

\usepackage{amsmath}
\usepackage{amssymb}
\usepackage{amsthm}
\usepackage{mathpartir}
\usepackage{cmll}
\usepackage{stmaryrd}
\usepackage{xspace}
\usepackage{color}
\usepackage{mathtools} 
\usepackage{tikz} 
\usetikzlibrary{arrows,positioning,fit,petri,decorations.markings}
\usepackage{subcaption} 

\usepackage{latexsym}
\usepackage{textcomp}
\usepackage{cmll} 
\usepackage{etoolbox} 
\usepackage{bbold} 
\usepackage{multicol}
\usepackage{paralist} 
\usepackage{wrapfig}

\newcommand{\etal}{~\textit{et al.}\xspace}

\newcommand{\Figure}[1]{Figure~\ref{fig:#1}}
\newcommand{\Section}[1]{Section~\ref{sec:#1}}
\newcommand{\Lemma}[1]{Lemma~\ref{lem:#1}}
\newcommand{\Proposition}[1]{Proposition~\ref{prop:#1}}
\newcommand{\Theorem}[1]{Theorem~\ref{thm:#1}}

\newtheorem{theorem}{Theorem}
\newtheorem{lemma}[theorem]{Lemma}
\newtheorem{proposition}[theorem]{Proposition}

\newtheorem{corollary}[theorem]{Corollary}

\newtheorem{definition}[theorem]{Definition}

\DeclareMathAlphabet{\mathcal}{OMS}{cmsy}{m}{n} 
\newcommand{\longcatfont}[1]{\textsc{#1}\xspace}
\newcommand{\catfont}[1]{\mathcal{#1}}
\newcommand{\functorfont}[1]{\mathrm{#1}\xspace}

\newcommand{\cat}[1]{\catfont{#1}}
\newcommand{\functor}[1]{\functorfont{#1}}

\newcommand{\op}[1]{{#1}^{op}}
\newcommand{\inv}[1]{#1^{-1}}

\newcommand{\F}{\ensuremath{\mathbb{F}}\xspace}

\newcommand{\shows}{\ensuremath{\vdash}}

\newcommand{\mprod}{\mathbin{\text{\footnotesize \ensuremath{\otimes}}}}
\newcommand{\aprod}{\mathbin{\text{\footnotesize \ensuremath{\&}}}}
\newcommand{\msum}{\mathbin{\text{\footnotesize \ensuremath{\parr}}}}
\newcommand{\asum}{\mathbin{\text{\footnotesize \ensuremath{\oplus}}}}
\newcommand{\lolto}{\ensuremath{\multimap}}

\newcommand{\dualize}[1]{\ensuremath{{#1^{\bot}}}}

\newif \ifcomments \commentstrue

\ifcomments
\newcommand{\jennifer}[1]{\textbf{\textcolor{red}{[ #1 --- Jennifer]}}}
\newcommand{\steve}[1]{\textbf{\textcolor{red}{[ #1 --- Steve]}}}

\newcommand{\fixme}[1]{\textbf{\textcolor{red}{[ Fixme: #1]}}}
\newcommand{\note}[1]{\textbf{\textcolor{blue}{[ Note: #1 ]}}}
\else

\newcommand{\jennifer}[1]{}
\newcommand{\steve}[1]{}

\newcommand{\fixme}[1]{}
\newcommand{\note}[1]{}
\fi

\newcommand{\Set}{\longcatfont{Set}}
\newcommand{\Rel}{\longcatfont{Rel}}
\newcommand{\FinVect}{\longcatfont{FinVect}}

\newcommand{\FinSet}{\longcatfont{FinSet}}

\newcommand{\FinBoolAlg}{\longcatfont{FinBoolAlg}}
\newcommand{\FinPoset}{\longcatfont{FinPoset}}
\newcommand{\FinLat}{\longcatfont{FinLat}}

\newcommand{\Hom}[1]{\ensuremath{\textrm{Hom}(#1)}}

\renewcommand{\star}[1]{{\ensuremath{#1^\ast}}}
\newcommand{\lowerstar}[1]{{\ensuremath{#1_\ast}}}

\newcommand{\mode}[1]{\mathsf{#1}}

\newcommand{\bangfunctor}[1]{\ensuremath{\lceil #1 \rceil}}
\newcommand{\whynotfunctor}[1]{\ensuremath{\lfloor #1 \rfloor}}

\newcommand{\interp}[1]{\llbracket #1 \rrbracket}

\newcommand{\ottdrule}[4][]{{\displaystyle\frac{\begin{array}{l}#2\end{array}}{#3}\quad\ottdrulename{#4}}}

\newcommand{\ottpremise}[1]{ #1 \\}
\newenvironment{ottdefnblock}[3][]{ \framebox{\mbox{#2}} \quad #3 \\[0pt]}{}

\newcommand{\ottnt}[1]{\mathit{#1}}

\newcommand{\ottkw}[1]{\mathbf{#1}}
\newcommand{\ottsym}[1]{#1}

\newcommand{\ottdrulename}[1]{\textsc{#1}}

\newcommand{\ottdrulelinXXAxiom}[1]{\ottdrule[#1]{%
}{
\ottnt{X}  \shows  \ottnt{X}}{%
{\ottdrulename{lin\_Axiom}}{}%
}}

\newcommand{\ottdrulelinXXMConjL}[1]{\ottdrule[#1]{%
\ottpremise{  \Gamma ,  \ottnt{X_{{\mathrm{1}}}} ^{ \mode{m} }   ,  \ottnt{X_{{\mathrm{2}}}} ^{ \mode{m} }    \shows  \Delta \quad  \mode{m}  \in \{   \mode{L}   ,   \mode{P}   \} }%
}{
 \Gamma ,  \ottsym{(}    \ottnt{X_{{\mathrm{1}}}}    \mprod    \ottnt{X_{{\mathrm{2}}}}    \ottsym{)} ^{ \mode{m} }    \shows  \Delta}{%
{\ottdrulename{lin\_MConjL}}{}%
}}

\newcommand{\ottdrulelinXXMConjR}[1]{\ottdrule[#1]{%
\ottpremise{\Gamma_{{\mathrm{1}}}  \shows   \Delta_{{\mathrm{1}}} ,  \ottnt{X_{{\mathrm{1}}}} ^{ \mode{m} }   \quad \Gamma_{{\mathrm{2}}}  \shows   \Delta_{{\mathrm{2}}} ,  \ottnt{X_{{\mathrm{2}}}} ^{ \mode{m} }   \quad  \mode{m}  \in \{   \mode{L}   ,   \mode{P}   \} }%
}{
 \Gamma_{{\mathrm{1}}} , \Gamma_{{\mathrm{2}}}   \shows    \Delta_{{\mathrm{1}}} , \Delta_{{\mathrm{2}}}  ,  \ottsym{(}    \ottnt{X_{{\mathrm{1}}}}    \mprod    \ottnt{X_{{\mathrm{2}}}}    \ottsym{)} ^{ \mode{m} }  }{%
{\ottdrulename{lin\_MConjR}}{}%
}}

\newcommand{\ottdrulelinXXOneL}[1]{\ottdrule[#1]{%
\ottpremise{\Gamma  \shows  \Delta \quad  \mode{m}  \in \{   \mode{L}   ,   \mode{P}   \} }%
}{
 \Gamma ,  1_{ \mode{m} }    \shows  \Delta}{%
{\ottdrulename{lin\_OneL}}{}%
}}

\newcommand{\ottdrulelinXXOneR}[1]{\ottdrule[#1]{%
\ottpremise{ \mode{m}  \in \{   \mode{L}   ,   \mode{P}   \} }%
}{
 \cdot   \shows   1_{ \mode{m} } }{%
{\ottdrulename{lin\_OneR}}{}%
}}

\newcommand{\ottdrulelinXXMDisjL}[1]{\ottdrule[#1]{%
\ottpremise{ \Gamma_{{\mathrm{1}}} ,  \ottnt{X_{{\mathrm{1}}}} ^{ \mode{m} }    \shows  \Delta_{{\mathrm{1}}} \quad  \Gamma_{{\mathrm{2}}} ,  \ottnt{X_{{\mathrm{2}}}} ^{ \mode{m} }    \shows  \Delta_{{\mathrm{2}}} \quad  \mode{m}  \in \{   \mode{L}   ,   \mode{C}   \} }%
}{
  \Gamma_{{\mathrm{1}}} , \Gamma_{{\mathrm{2}}}  ,  \ottsym{(}    \ottnt{X_{{\mathrm{1}}}}    \msum    \ottnt{X_{{\mathrm{2}}}}    \ottsym{)} ^{ \mode{m} }    \shows   \Delta_{{\mathrm{1}}} , \Delta_{{\mathrm{2}}} }{%
{\ottdrulename{lin\_MDisjL}}{}%
}}

\newcommand{\ottdrulelinXXMDisjR}[1]{\ottdrule[#1]{%
\ottpremise{\Gamma  \shows    \Delta ,  \ottnt{X_{{\mathrm{1}}}} ^{ \mode{m} }   ,  \ottnt{X_{{\mathrm{2}}}} ^{ \mode{m} }   \quad  \mode{m}  \in \{   \mode{L}   ,   \mode{C}   \} }%
}{
\Gamma  \shows   \Delta ,  \ottsym{(}    \ottnt{X_{{\mathrm{1}}}}    \msum    \ottnt{X_{{\mathrm{2}}}}    \ottsym{)} ^{ \mode{m} }  }{%
{\ottdrulename{lin\_MDisjR}}{}%
}}

\newcommand{\ottdrulelinXXBotL}[1]{\ottdrule[#1]{%
\ottpremise{ \mode{m}  \in \{   \mode{L}   ,   \mode{C}   \} }%
}{
 \bot_{ \mode{m} }   \shows   \cdot }{%
{\ottdrulename{lin\_BotL}}{}%
}}

\newcommand{\ottdrulelinXXBotR}[1]{\ottdrule[#1]{%
\ottpremise{\Gamma  \shows  \Delta \quad  \mode{m}  \in \{   \mode{L}   ,   \mode{C}   \} }%
}{
\Gamma  \shows   \Delta ,  \bot_{ \mode{m} }  }{%
{\ottdrulename{lin\_BotR}}{}%
}}

\newcommand{\ottdrulelinXXAConjLOne}[1]{\ottdrule[#1]{%
\ottpremise{ \Gamma , \ottnt{A}   \shows  \Delta}%
}{
 \Gamma , \ottnt{A}  \aprod  \ottnt{B}   \shows  \Delta}{%
{\ottdrulename{lin\_AConjL1}}{}%
}}

\newcommand{\ottdrulelinXXAConjLTwo}[1]{\ottdrule[#1]{%
\ottpremise{ \Gamma , \ottnt{B}   \shows  \Delta}%
}{
 \Gamma , \ottnt{A}  \aprod  \ottnt{B}   \shows  \Delta}{%
{\ottdrulename{lin\_AConjL2}}{}%
}}

\newcommand{\ottdrulelinXXAConjR}[1]{\ottdrule[#1]{%
\ottpremise{\Gamma  \shows   \Delta , \ottnt{A}  \quad \Gamma  \shows   \Delta , \ottnt{B} }%
}{
\Gamma  \shows   \Delta , \ottnt{A}  \aprod  \ottnt{B} }{%
{\ottdrulename{lin\_AConjR}}{}%
}}

\newcommand{\ottdrulelinXXTopR}[1]{\ottdrule[#1]{%
}{
\Gamma  \shows   \Delta ,  \top  }{%
{\ottdrulename{lin\_TopR}}{}%
}}

\newcommand{\ottdrulelinXXADisjL}[1]{\ottdrule[#1]{%
\ottpremise{ \Gamma , \ottnt{A}   \shows  \Delta \quad  \Gamma , \ottnt{B}   \shows  \Delta}%
}{
 \Gamma , \ottnt{A}  \asum  \ottnt{B}   \shows  \Delta}{%
{\ottdrulename{lin\_ADisjL}}{}%
}}

\newcommand{\ottdrulelinXXADisjROne}[1]{\ottdrule[#1]{%
\ottpremise{\Gamma  \shows   \Delta , \ottnt{A} }%
}{
\Gamma  \shows   \Delta , \ottnt{A}  \asum  \ottnt{B} }{%
{\ottdrulename{lin\_ADisjR1}}{}%
}}

\newcommand{\ottdrulelinXXADisjRTwo}[1]{\ottdrule[#1]{%
\ottpremise{\Gamma  \shows   \Delta , \ottnt{B} }%
}{
\Gamma  \shows   \Delta , \ottnt{A}  \asum  \ottnt{B} }{%
{\ottdrulename{lin\_ADisjR2}}{}%
}}

\newcommand{\ottdrulelinXXZeroL}[1]{\ottdrule[#1]{%
}{
 \Gamma , \ottsym{0}   \shows  \Delta}{%
{\ottdrulename{lin\_0L}}{}%
}}

\newcommand{\ottdrulelinXXWeakeningL}[1]{\ottdrule[#1]{%
\ottpremise{\Gamma  \shows  \Delta}%
}{
 \Gamma , \ottnt{P}   \shows  \Delta}{%
{\ottdrulename{lin\_WeakeningL}}{}%
}}

\newcommand{\ottdrulelinXXContractionL}[1]{\ottdrule[#1]{%
\ottpremise{  \Gamma , \ottnt{P}  , \ottnt{P}   \shows  \Delta}%
}{
 \Gamma , \ottnt{P}   \shows  \Delta}{%
{\ottdrulename{lin\_ContractionL}}{}%
}}

\newcommand{\ottdrulelinXXFBangL}[1]{\ottdrule[#1]{%
\ottpremise{ \Gamma , \ottnt{P}   \shows  \Delta}%
}{
 \Gamma ,  F_\oc  \, \ottnt{P}   \shows  \Delta}{%
{\ottdrulename{lin\_FBangL}}{}%
}}

\newcommand{\ottdrulelinXXFBangR}[1]{\ottdrule[#1]{%
\ottpremise{ \Gamma ^{\mode P}   \Vdash    \Delta ^{\mode C}  , \ottnt{P} }%
}{
 \Gamma ^{\mode P}   \shows    \Delta ^{\mode C}  ,  F_\oc  \, \ottnt{P} }{%
{\ottdrulename{lin\_FBangR}}{}%
}}

\newcommand{\ottdrulelinXXFWhynotL}[1]{\ottdrule[#1]{%
\ottpremise{  \Gamma ^{\mode P}  , \ottnt{C}   \Vdash   \Delta ^{\mode C} }%
}{
  \Gamma ^{\mode P}  ,  F_\wn  \, \ottnt{C}   \shows   \Delta ^{\mode C} }{%
{\ottdrulename{lin\_FWhynotL}}{}%
}}

\newcommand{\ottdrulelinXXFWhynotR}[1]{\ottdrule[#1]{%
\ottpremise{\Gamma  \shows   \Delta , \ottnt{C} }%
}{
\Gamma  \shows   \Delta ,  F_\wn  \, \ottnt{C} }{%
{\ottdrulename{lin\_FWhynotR}}{}%
}}

\newcommand{\ottdrulelinXXBangL}[1]{\ottdrule[#1]{%
\ottpremise{ \Gamma , \ottnt{A}   \shows  \Delta}%
}{
 \Gamma ,  \bangfunctor{ \ottnt{A} }    \shows  \Delta}{%
{\ottdrulename{lin\_BangL}}{}%
}}

\newcommand{\ottdrulelinXXWhynotR}[1]{\ottdrule[#1]{%
\ottpremise{\Gamma  \shows   \Delta , \ottnt{A} }%
}{
\Gamma  \shows   \Delta ,  \whynotfunctor{ \ottnt{A} }  }{%
{\ottdrulename{lin\_WhynotR}}{}%
}}

\newcommand{\ottdrulelinXXCutLinear}[1]{\ottdrule[#1]{%
\ottpremise{\Gamma_{{\mathrm{1}}}  \shows   \Delta_{{\mathrm{1}}} , \ottnt{A}  \quad  \ottnt{A} , \Gamma_{{\mathrm{2}}}   \shows  \Delta_{{\mathrm{2}}}}%
}{
 \Gamma_{{\mathrm{1}}} , \Gamma_{{\mathrm{2}}}   \shows   \Delta_{{\mathrm{1}}} , \Delta_{{\mathrm{2}}} }{%
{\ottdrulename{lin\_CutLinear}}{}%
}}

\newcommand{\ottdrulelinXXCutProducer}[1]{\ottdrule[#1]{%
\ottpremise{ \Gamma_{{\mathrm{1}}} ^{\mode P}   \Vdash    \Delta_{{\mathrm{1}}} ^{\mode C}  , \ottnt{P}  \quad  \ottnt{P} , \Gamma_{{\mathrm{2}}}   \shows  \Delta_{{\mathrm{2}}}}%
}{
  \Gamma_{{\mathrm{1}}} ^{\mode P}  , \Gamma_{{\mathrm{2}}}   \shows    \Delta_{{\mathrm{1}}} ^{\mode C}  , \Delta_{{\mathrm{2}}} }{%
{\ottdrulename{lin\_CutProducer}}{}%
}}

\newcommand{\ottdrulelinXXCutConsumer}[1]{\ottdrule[#1]{%
\ottpremise{\Gamma_{{\mathrm{1}}}  \shows   \Delta_{{\mathrm{1}}} , \ottnt{C}  \quad   \ottnt{C} , \Gamma_{{\mathrm{2}}}  ^{\mode P}   \Vdash   \Delta_{{\mathrm{2}}} ^{\mode C} }%
}{
  \Gamma_{{\mathrm{1}}} , \Gamma_{{\mathrm{2}}}  ^{\mode P}   \shows    \Delta_{{\mathrm{1}}} , \Delta_{{\mathrm{2}}}  ^{\mode C} }{%
{\ottdrulename{lin\_CutConsumer}}{}%
}}

\newcommand{\ottdrulelinXXCutProducerPlus}[1]{\ottdrule[#1]{%
\ottpremise{ \Gamma_{{\mathrm{1}}} ^{\mode P}   \Vdash    \Delta_{{\mathrm{1}}} ^{\mode C}  , \ottnt{P}  \quad   ( \ottnt{P} )_{ n }  , \Gamma_{{\mathrm{2}}}   \shows  \Delta_{{\mathrm{2}}}}%
}{
  \Gamma_{{\mathrm{1}}} ^{\mode P}  , \Gamma_{{\mathrm{2}}}   \shows    \Delta_{{\mathrm{1}}} ^{\mode C}  , \Delta_{{\mathrm{2}}} }{%
{\ottdrulename{lin\_CutProducerPlus}}{}%
}}

\newcommand{\ottdrulelinXXCutConsumerPlus}[1]{\ottdrule[#1]{%
\ottpremise{\Gamma_{{\mathrm{1}}}  \shows   \Delta_{{\mathrm{1}}} ,  ( \ottnt{C} )_{ n }   \quad   \ottnt{C} , \Gamma_{{\mathrm{2}}}  ^{\mode P}   \Vdash   \Delta_{{\mathrm{2}}} ^{\mode C} }%
}{
  \Gamma_{{\mathrm{1}}} , \Gamma_{{\mathrm{2}}}  ^{\mode P}   \shows    \Delta_{{\mathrm{1}}} , \Delta_{{\mathrm{2}}}  ^{\mode C} }{%
{\ottdrulename{lin\_CutConsumerPlus}}{}%
}}

\newcommand{\ottdrulelinXXLinDualL}[1]{\ottdrule[#1]{%
\ottpremise{\Gamma  \shows   \Delta , \ottnt{A} }%
}{
 \Gamma ,  \dualize{ \ottnt{A} }    \shows  \Delta}{%
{\ottdrulename{lin\_LinDualL}}{}%
}}

\newcommand{\ottdrulelinXXLinDualR}[1]{\ottdrule[#1]{%
\ottpremise{ \Gamma , \ottnt{A}   \shows  \Delta}%
}{
\Gamma  \shows   \Delta ,  \dualize{ \ottnt{A} }  }{%
{\ottdrulename{lin\_LinDualR}}{}%
}}

\newcommand{\ottdrulelinXXProdDualL}[1]{\ottdrule[#1]{%
\ottpremise{\Gamma  \shows   \Delta , \ottnt{P} }%
}{
 \Gamma ,  \star{ \ottnt{P} }    \shows  \Delta}{%
{\ottdrulename{lin\_ProdDualL}}{}%
}}

\newcommand{\ottdrulelinXXConsDualL}[1]{\ottdrule[#1]{%
\ottpremise{\Gamma  \shows   \Delta , \ottnt{C} }%
}{
 \Gamma ,  \lowerstar{ { \ottnt{C} } }    \shows  \Delta}{%
{\ottdrulename{lin\_ConsDualL}}{}%
}}

\newcommand{\ottdrulepersXXPAxiom}[1]{\ottdrule[#1]{%
}{
\ottnt{P}  \Vdash  \ottnt{P}}{%
{\ottdrulename{pers\_PAxiom}}{}%
}}

\newcommand{\ottdrulepersXXCAxiom}[1]{\ottdrule[#1]{%
}{
\ottnt{C}  \Vdash  \ottnt{C}}{%
{\ottdrulename{pers\_CAxiom}}{}%
}}

\newcommand{\ottdrulepersXXMConjL}[1]{\ottdrule[#1]{%
\ottpremise{  \Gamma , \ottnt{P_{{\mathrm{1}}}}  , \ottnt{P_{{\mathrm{2}}}}   \Vdash  \Delta}%
}{
 \Gamma , \ottnt{P_{{\mathrm{1}}}}  \mprod  \ottnt{P_{{\mathrm{2}}}}   \Vdash  \Delta}{%
{\ottdrulename{pers\_MConjL}}{}%
}}

\newcommand{\ottdrulepersXXMConjR}[1]{\ottdrule[#1]{%
\ottpremise{\Gamma_{{\mathrm{1}}}  \Vdash   \Delta_{{\mathrm{1}}} , \ottnt{P_{{\mathrm{1}}}}  \quad \Gamma_{{\mathrm{2}}}  \Vdash   \Delta_{{\mathrm{2}}} , \ottnt{P_{{\mathrm{2}}}} }%
}{
 \Gamma_{{\mathrm{1}}} , \Gamma_{{\mathrm{2}}}   \Vdash    \Delta_{{\mathrm{1}}} , \Delta_{{\mathrm{2}}}  , \ottnt{P_{{\mathrm{1}}}}  \mprod  \ottnt{P_{{\mathrm{2}}}} }{%
{\ottdrulename{pers\_MConjR}}{}%
}}

\newcommand{\ottdrulepersXXPOneL}[1]{\ottdrule[#1]{%
\ottpremise{\Gamma  \Vdash  \Delta}%
}{
 \Gamma ,  1_{\mode P}    \Vdash  \Delta}{%
{\ottdrulename{pers\_POneL}}{}%
}}

\newcommand{\ottdrulepersXXPOneR}[1]{\ottdrule[#1]{%
}{
 \cdot   \Vdash   1_{\mode P} }{%
{\ottdrulename{pers\_POneR}}{}%
}}

\newcommand{\ottdrulepersXXMDisjL}[1]{\ottdrule[#1]{%
\ottpremise{ \Gamma_{{\mathrm{1}}} , \ottnt{C_{{\mathrm{1}}}}   \Vdash  \Delta_{{\mathrm{1}}} \quad  \Gamma_{{\mathrm{2}}} , \ottnt{C_{{\mathrm{2}}}}   \Vdash  \Delta_{{\mathrm{2}}}}%
}{
  \Gamma_{{\mathrm{1}}} , \Gamma_{{\mathrm{2}}}  , \ottnt{C_{{\mathrm{1}}}}  \msum  \ottnt{C_{{\mathrm{2}}}}   \Vdash   \Delta_{{\mathrm{1}}} , \Delta_{{\mathrm{2}}} }{%
{\ottdrulename{pers\_MDisjL}}{}%
}}

\newcommand{\ottdrulepersXXMDisjR}[1]{\ottdrule[#1]{%
\ottpremise{\Gamma  \Vdash    \Delta , \ottnt{C_{{\mathrm{1}}}}  , \ottnt{C_{{\mathrm{2}}}} }%
}{
\Gamma  \Vdash   \Delta , \ottnt{C_{{\mathrm{1}}}}  \msum  \ottnt{C_{{\mathrm{2}}}} }{%
{\ottdrulename{pers\_MDisjR}}{}%
}}

\newcommand{\ottdrulepersXXCBotL}[1]{\ottdrule[#1]{%
}{
 \bot_{\mode C}   \Vdash   \cdot }{%
{\ottdrulename{pers\_CBotL}}{}%
}}

\newcommand{\ottdrulepersXXCBotR}[1]{\ottdrule[#1]{%
\ottpremise{\Gamma  \Vdash  \Delta}%
}{
\Gamma  \Vdash   \Delta ,  \bot_{\mode C}  }{%
{\ottdrulename{pers\_CBotR}}{}%
}}

\newcommand{\ottdrulepersXXBangR}[1]{\ottdrule[#1]{%
\ottpremise{ \Gamma ^{\mode P}   \shows    \Delta ^{\mode C}  , \ottnt{A} }%
}{
 \Gamma ^{\mode P}   \Vdash    \Delta ^{\mode C}  ,  \bangfunctor{ \ottnt{A} }  }{%
{\ottdrulename{pers\_BangR}}{}%
}}

\newcommand{\ottdrulepersXXWhynotL}[1]{\ottdrule[#1]{%
\ottpremise{  \Gamma ^{\mode P}  , \ottnt{A}   \shows   \Delta ^{\mode C} }%
}{
  \Gamma ^{\mode P}  ,  \whynotfunctor{ \ottnt{A} }    \Vdash   \Delta ^{\mode C} }{%
{\ottdrulename{pers\_WhynotL}}{}%
}}

\newcommand{\ottdrulepersXXWeakeningL}[1]{\ottdrule[#1]{%
\ottpremise{\Gamma  \Vdash  \Delta}%
}{
 \Gamma , \ottnt{P}   \Vdash  \Delta}{%
{\ottdrulename{pers\_WeakeningL}}{}%
}}

\newcommand{\ottdrulepersXXContractionL}[1]{\ottdrule[#1]{%
\ottpremise{  \Gamma , \ottnt{P}  , \ottnt{P}   \Vdash  \Delta}%
}{
 \Gamma , \ottnt{P}   \Vdash  \Delta}{%
{\ottdrulename{pers\_ContractionL}}{}%
}}

\newcommand{\ottdrulepersXXCutProducer}[1]{\ottdrule[#1]{%
\ottpremise{ \Gamma_{{\mathrm{1}}} ^{\mode P}   \Vdash    \Delta_{{\mathrm{1}}} ^{\mode C}  , \ottnt{P}  \quad  \ottnt{P} , \Gamma_{{\mathrm{2}}}   \Vdash  \Delta_{{\mathrm{2}}}}%
}{
  \Gamma_{{\mathrm{1}}} ^{\mode P}  , \Gamma_{{\mathrm{2}}}   \Vdash    \Delta_{{\mathrm{1}}} ^{\mode C}  , \Delta_{{\mathrm{2}}} }{%
{\ottdrulename{pers\_CutProducer}}{}%
}}

\newcommand{\ottdrulepersXXCutConsumer}[1]{\ottdrule[#1]{%
\ottpremise{\Gamma_{{\mathrm{1}}}  \Vdash   \Delta_{{\mathrm{1}}} , \ottnt{C}  \quad   \ottnt{C} , \Gamma_{{\mathrm{2}}}  ^{\mode P}   \Vdash   \Delta_{{\mathrm{2}}} ^{\mode C} }%
}{
  \Gamma_{{\mathrm{1}}} , \Gamma_{{\mathrm{2}}}  ^{\mode P}   \Vdash    \Delta_{{\mathrm{1}}} , \Delta_{{\mathrm{2}}}  ^{\mode C} }{%
{\ottdrulename{pers\_CutConsumer}}{}%
}}

\newcommand{\ottdrulepersXXCutProducerPlus}[1]{\ottdrule[#1]{%
\ottpremise{ \Gamma_{{\mathrm{1}}} ^{\mode P}   \Vdash    \Delta_{{\mathrm{1}}} ^{\mode C}  , \ottnt{P}  \quad   ( \ottnt{P} )_{ n }  , \Gamma_{{\mathrm{2}}}   \Vdash  \Delta_{{\mathrm{2}}}}%
}{
  \Gamma_{{\mathrm{1}}} ^{\mode P}  , \Gamma_{{\mathrm{2}}}   \Vdash    \Delta_{{\mathrm{1}}} ^{\mode C}  , \Delta_{{\mathrm{2}}} }{%
{\ottdrulename{pers\_CutProducerPlus}}{}%
}}

\newcommand{\ottdrulepersXXCutConsumerPlus}[1]{\ottdrule[#1]{%
\ottpremise{\Gamma_{{\mathrm{1}}}  \Vdash   \Delta_{{\mathrm{1}}} ,  ( \ottnt{C} )_{ n }   \quad   \ottnt{C} , \Gamma_{{\mathrm{2}}}  ^{\mode P}   \Vdash   \Delta_{{\mathrm{2}}} ^{\mode C} }%
}{
  \Gamma_{{\mathrm{1}}} , \Gamma_{{\mathrm{2}}}  ^{\mode P}   \Vdash    \Delta_{{\mathrm{1}}} , \Delta_{{\mathrm{2}}}  ^{\mode C} }{%
{\ottdrulename{pers\_CutConsumerPlus}}{}%
}}

\newcommand{\ottdrulepersXXProdDualL}[1]{\ottdrule[#1]{%
\ottpremise{\Gamma  \Vdash   \Delta , \ottnt{P} }%
}{
 \Gamma ,  \star{ \ottnt{P} }    \Vdash  \Delta}{%
{\ottdrulename{pers\_ProdDualL}}{}%
}}

\newcommand{\ottdrulepersXXConsDualL}[1]{\ottdrule[#1]{%
\ottpremise{\Gamma  \Vdash   \Delta , \ottnt{C} }%
}{
 \Gamma ,  \lowerstar{ { \ottnt{C} } }    \Vdash  \Delta}{%
{\ottdrulename{pers\_ConsDualL}}{}%
}}

\renewcommand{\ottdrule}[4][]{{\displaystyle\frac{\begin{array}{l}#2\end{array}}{#3}\enskip\ottdrulename{#1}}}

\title{A Linear/Producer/Consumer Model of Classical Linear Logic}

\author{Jennifer Paykin \qquad\qquad Steve Zdancewic
\institute{University of Pennsylvania\\ Philadelphia, USA}
\email{jpaykin@seas.upenn.edu \qquad\qquad stevez@cis.upenn.edu}
}
\def\authorrunning{J. Paykin \& S. Zdancewic}
\def\titlerunning{An LPC Model of Classical Linear Logic}

\begin{document}

\setlength\abovedisplayskip{6pt plus 2pt minus 2pt}%
\setlength\abovedisplayshortskip{4pt plus 2pt minus 2pt}%
\setlength\belowdisplayskip{6pt plus 2pt minus 2pt}%
\setlength\belowdisplayshortskip{4pt plus 2pt minus 2pt}%

\maketitle

\begin{abstract}
This paper defines a new proof- and category-theoretic framework for
\textit{classical linear logic} that separates reasoning into one linear regime
and two persistent regimes corresponding to ! and ?. The resulting
linear/producer/consumer (LPC) logic puts the three classes of
propositions on the same semantic footing, following Benton's
linear/non-linear formulation of intuitionistic linear logic.
Semantically, LPC corresponds to a system of three categories
connected by adjunctions reflecting the linear/producer/consumer
structure. The paper's metatheoretic results include admissibility
theorems for the cut and duality rules, and a translation of the LPC
logic into category theory. The work also presents several
concrete instances of the LPC model.

\end{abstract}

\section{Introduction}
\label{sec:intro}
\begin{wrapfigure}{R}{0.35\textwidth}
\vspace{-8mm}
\begin{center}\scalebox{0.9}{\begin{tikzpicture}
    \node (L) [circle,draw=black,minimum size=1cm] at ( 2,0) {$ \cat{L} $} ;
    \node (P) [circle,draw=black,minimum size=1cm] at (-2,0) {$ \cat{P} $} ;
    \node (perp) at (0,0) {$\top$};

    \draw [bend right=15,-triangle 90] (L) to node [auto,swap] {$\functor{G}$} (P) ;
    \draw [bend right=15,-triangle 90] (P) to node [auto,swap] {$\functor{F}$} (L) ;
\end{tikzpicture}}\end{center}
\vspace{-1em}
\caption{Categorical model of ILL with linear and persistent categories.}
\label{fig:2cat}
\end{wrapfigure}
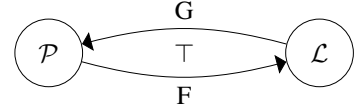

Since its introduction by Girard in 1987, linear logic has
been found to have a range of applications in logic, proof theory, and
programming languages. Its power stems from its ability to carefully manage
resource usage: it makes a crucial distinction between \textit{linear}
(used exactly once) and \textit{persistent} (unrestricted use)
hypotheses, internalizing the latter via the $ \oc $ connective.  From a
semantic point of view, the literature has converged (following
Benton~\cite{Benton94mixed}) on an
interpretation of $\oc$ as a comonad given by \( \oc = \functor{F} \circ \functor{G}
\) where $\functor{F} \dashv \functor{G}$ is a symmetric monoidal adjunction between
categories $ \cat{L} $ and $ \cat{P} $ arranged as shown in \Figure{2cat}.

Here, $ \cat{L} $ (for ``linear'') is a symmetric monoidal closed
category and $ \cat{P} $ (for ``persistent'') is
a cartesian category.  This is, by now, a standard way of interpreting
\textit{intuitionistic} linear logic (for details, see
Melli{\`e}s~\cite{Mel09}).  
If, in addition, the category $ \cat{L} $ is $*$-autonomous, the
structure above is sufficient to interpret \textit{classical} linear
logic, where the monad $\wn$ is determined by \( \wn =  \dualize{ \ottsym{(}   \op{ \functor{F} }  \, \ottsym{(}   \op{ \functor{G} }  \, \ottsym{(}   \dualize{ \ottsym{-} }   \ottsym{)}  \ottsym{)}  \ottsym{)} }  \).  While sound, this situation unnecessarily commits
to a particular implementation of $\wn$ in term of $ \op{  \cat{P}  } $.
The $ \textsc{LPC} $ framework absolves us of this commitment by opening up
a new range of semantic models, discussed in \Section{examples}.

With that motivation, this paper defines a proof- and
category-theoretic framework for full \textit{classical linear logic}
that uses \textit{two} persistent categories: one corresponding to $\oc$ and one
to $\wn$.  The resulting categorical structure is shown in
Figure~\ref{fig:3cat}, where $ \cat{P} $ takes the place of the
``producing'' category, in duality with $ \cat{C} $ as the ``consuming''
category. This terminology comes from the observations that:
\[ \begin{array}{ccc}
    \oc  \ottnt{A}  \shows   1  &\qquad\qquad\qquad&  \bot   \shows  \wn  \ottnt{A} \\
    \oc  \ottnt{A}  \shows  \ottnt{A} &\qquad\qquad\qquad& \ottnt{A}  \shows  \wn  \ottnt{A} \\
    \oc  \ottnt{A}  \shows  \oc  \ottnt{A}  \mprod  \oc  \ottnt{A} &\qquad\qquad\qquad& \wn  \ottnt{A}  \msum  \wn  \ottnt{A}  \shows  \wn  \ottnt{A}
\end{array} \]
%
%
Intuitively, the left group means that $\oc  \ottnt{A}$ is sufficient to \textit{produce} any 
number of copies of $\ottnt{A}$ and, dually, the right group says that
$\wn  \ottnt{A}$ can \textit{consume} any number of copies of $\ottnt{A}$.


%


\section{$ \textsc{LPC} $ Logic}
\label{sec:logic}
\begin{wrapfigure}[11]{R}{0.45\textwidth}
\small
\vspace{-8mm}
\centering{
\begin{tikzpicture}
    \node (L) [circle,draw=black,minimum size=0.5cm] at ( 0, 2) {$ \cat{L} $} ;
    \node (P) [circle,draw=black,minimum size=0.5cm] at (-2,-1) {$ \cat{P} $} ;
    \node (C) [circle,draw=black,minimum size=0.5cm] at ( 2,-1) {$ \cat{C} $} ;
    \node [rotate=145] at (-1,0.5) {$ \dashv $} ;
    \node [rotate=215] at ( 1,0.5) {$ \dashv $} ;

    \draw [bend right=15,-triangle 90] (L) to node [auto,swap] {\small $ \bangfunctor{ \ottsym{-} } $} (P) ;
    \draw [bend left=15 ,-triangle 90] (L) to node [auto     ] {\small $ \whynotfunctor{ \ottsym{-} } $} (C) ;

    \draw [bend right=15,-triangle 90] (P) to node [auto,swap] {\small $ F_\oc $} (L) ;
    \draw [bend left=15 ,-triangle 90] (C) to node [auto     ] {\small $ F_\wn $} (L) ;

    \draw [loop above, looseness=8, -triangle 90] (L) to node [auto] {\small $ \dualize{ \ottsym{(}  \ottsym{-}  \ottsym{)} } $} (L) ;
    \draw [bend right=12,-triangle 90] (P) to node [auto,swap] {\small $ \lowerstar{ \ottsym{(}  \ottsym{-}  \ottsym{)} } $} (C) ;
    \draw [bend right=12,-triangle 90] (C) to node [auto,swap] {\small $ \star{ \ottsym{(}  \ottsym{-}  \ottsym{)} } $} (P) ;
\end{tikzpicture}
}
\vspace{-1mm}
\caption{Categorical model of classical linear logic with linear, producing and consuming categories.}
\label{fig:3cat} 
\end{wrapfigure}
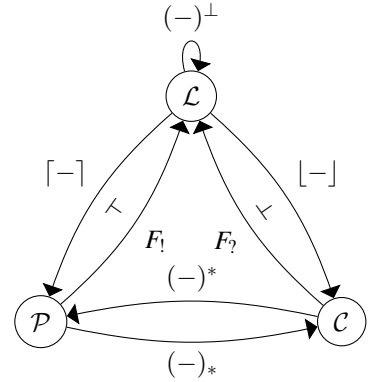

The syntax of the $ \textsc{LPC} $ logic is made up of three
syntactic forms for propositions: linear propositions $\ottnt{A}$,
producer propositions $\ottnt{P}$, and consumer propositions $\ottnt{C}$.
%
\[\begin{array}{ccccccccc}
    \ottnt{A} & ::= &  \top  & \mid & \ottnt{A_{{\mathrm{1}}}}  \aprod  \ottnt{A_{{\mathrm{2}}}}
          & \mid & \ottsym{0} & \mid & \ottnt{A_{{\mathrm{1}}}}  \asum  \ottnt{A_{{\mathrm{2}}}} \\
          & \mid &  1_{\mode L}  & \mid & \ottnt{A_{{\mathrm{1}}}}  \mprod  \ottnt{A_{{\mathrm{2}}}}
          & \mid &  \bot_{\mode L}  & \mid & \ottnt{A_{{\mathrm{1}}}}  \msum  \ottnt{A_{{\mathrm{2}}}} \\
          & \mid &  F_\oc  \, \ottnt{P} & \mid &  F_\wn  \, \ottnt{C} \\
    \ottnt{P} & ::= &  1_{\mode P}  & \mid & \ottnt{P_{{\mathrm{1}}}}  \mprod  \ottnt{P_{{\mathrm{2}}}} & \mid &  \bangfunctor{ \ottnt{A} }   \\
    \ottnt{C} & ::= &  \bot_{\mode C}  & \mid & \ottnt{C_{{\mathrm{1}}}}  \msum  \ottnt{C_{{\mathrm{2}}}} & \mid &  \whynotfunctor{ \ottnt{A} } 
\end{array}\]
The syntactic form of a proposition is called its \emph{mode}---linear
$ \mode{L} $, producing $ \mode{P} $ or consuming $ \mode{C} $. The meta-variable
$\ottnt{X}$ ranges over propositions of any mode, and the tagged
meta-variable $ \ottnt{X} ^{ \mode{m} } $ ranges over propositions of mode $\mode{m}$.  The
term \emph{persistent} refers to either producer or consumer propositions.

$ \textsc{LPC} $ replaces the usual constructors $ \oc $ and $ \wn $ with two pairs of connectives:
$ F_\oc $ and $ \bangfunctor{ \ottsym{-} } $ for $ \oc $ and  $ F_\wn $ and $ \whynotfunctor{ \ottsym{-} } $ for $ \wn $.
If $\ottnt{A}$ is a linear proposition, $ \bangfunctor{ \ottnt{A} } $ is a producer
and $ \whynotfunctor{ \ottnt{A} } $ is a consumer. On the other hand, a producer proposition
$\ottnt{P}$ may be ``frozen'' into a linear proposition $ F_\oc  \, \ottnt{P}$, effectively discarding its
persistent characteristics. Similarly for a consumer $\ottnt{C}$,
$ F_\wn  \, \ottnt{C}$ is linear. The linear propositions $\oc  \ottnt{A}$ and $\wn  \ottnt{A}$ are encoded
in this system as $ F_\oc  \, \ottsym{(}   \bangfunctor{ \ottnt{A} }   \ottsym{)}$ and $ F_\wn  \, \ottsym{(}   \whynotfunctor{ \ottnt{A} }   \ottsym{)}$ respectively. 


The inference rules of the logic are shown in 
Figures~\ref{fig:linrules-linear} and~\ref{fig:auxrules}.  There are two
judgments: the linear sequent $\Gamma  \shows  \Delta$ and the
persistent sequent $\Gamma  \Vdash  \Delta$.  In the linear sequent, the
(unordered) contexts $\Gamma$ and $\Delta$ may be made up of
propositions of any mode; in the persistent sequent, the
contexts may contain only persistent propositions.  
The meta-variable $ \Gamma ^{\mode P} $ refers to contexts made up entirely 
of producer propositions, and $ \Delta ^{\mode C} $ refers to contexts
of consumer propositions.

\begin{figure}[htb]
\small
\[
\begin{array}{cccc}

    \multicolumn{2}{c}{
    \ottdrulelinXXAxiom{$  \textsc{Ax}  ^{\vdash} $}}
    &\ottdrulelinXXTopR{$    \top  _{  \mode{L}  }  ^{\vdash}  \text{-R} $}
    &\ottdrulelinXXZeroL{$    0  _{  \mode{L}  }  ^{\vdash}  \text{-L} $}
    \\ \\
    \ottdrulelinXXAConjLOne{$    \aprod  _{  \mode{L}  }  ^{\vdash}  \text{-L}   \ottsym{1}$}
    &\ottdrulelinXXAConjLTwo{$    \aprod  _{  \mode{L}  }  ^{\vdash}  \text{-L}   \ottsym{2}$}
    &
    \ottdrulelinXXADisjROne{$    \asum  _{  \mode{L}  }  ^{\vdash}  \text{-R}   \ottsym{1}$}
    &\ottdrulelinXXADisjRTwo{$    \asum  _{  \mode{L}  }  ^{\vdash}  \text{-R}   \ottsym{2}$}
    \\ \\
    \multicolumn{2}{c}{\ottdrulelinXXAConjR{$    \aprod  _{  \mode{L}  }  ^{\vdash}  \text{-R} $}}
    &\multicolumn{2}{c}{\ottdrulelinXXADisjL{$    \asum  _{  \mode{L}  }  ^{\vdash}  \text{-L} $}}
    \\ \\
    \multicolumn{4}{c}{
    \ottdrulelinXXMConjL{$   \mprod  ^{\vdash}  \text{-L} $}
    \qquad\qquad
    \ottdrulelinXXMConjR{$   \mprod  ^{\vdash}  \text{-R} $}}
    \\ \\
    \multicolumn{2}{c}{\ottdrulelinXXOneL{$   1  ^{\vdash}  \text{-L} $}}
    &
    \multicolumn{2}{c}{\ottdrulelinXXOneR{$   1  ^{\vdash}  \text{-R} $}}
    \\ \\
    \multicolumn{4}{c}{
    \ottdrulelinXXMDisjL{$   \msum  ^{\vdash}  \text{-L} $}
    \qquad\qquad
    \ottdrulelinXXMDisjR{$   \msum  ^{\vdash}  \text{-R} $}}
    \\ \\
    \multicolumn{2}{c}{\ottdrulelinXXBotL{$   \bot  ^{\vdash}  \text{-L} $}}
    &
    \multicolumn{2}{c}{\ottdrulelinXXBotR{$   \bot  ^{\vdash}  \text{-R} $}}
\end{array}
\]
\vspace{-2em}
\caption{Inference Rules for Linear Sequent}
\label{fig:linrules-linear}
\end{figure}

The linear inference rules in 
Figures \ref{fig:linrules-linear} and \ref{fig:linrules-persistent}
encompass rules for the units and the linear operators
$ \asum $, $ \aprod $, $ \mprod $ and $ \msum $. It is worth noting
that the multiplicative product $ \mprod $ is defined only on linear and producer propositions,
while the multiplicative sum $ \msum $ is defined only on linear and consumer 
propositions.\footnote{The persistent operators in this paper are necessary
for $ \textsc{LPC} $, but in general are not restricted to the sum and product. 
Other operators, like $\rightarrow$ or $\vee$, could be incorporated
so long as every producer operator has a dual for consumers.}

Weakening and contraction can be applied for producers on the left-hand-side
and consumers on the right-hand-side of both the linear and persistent sequents.
For producers, that is:
\[  \ottdrulelinXXWeakeningL{$   \textsc{W}  ^{\vdash}  \text{-L} $}
    \qquad
    \ottdrulepersXXWeakeningL{$   \textsc{W}  ^{\Vdash}  \text{-L} $}
    \qquad
    \ottdrulelinXXContractionL{$   \textsc{C}  ^{\vdash}  \text{-L} $}
    \qquad
    \ottdrulepersXXContractionL{$   \textsc{C}  ^{\Vdash}  \text{-L} $}
\]

The rules for the operators $ F_\oc $, $ F_\wn $, $ \bangfunctor{ \ottsym{-} } $ and $ \whynotfunctor{ \ottsym{-} } $
are given in \Figure{structrules}. These rules encode dereliction and promotion 
for $ \oc $ and $ \wn $ by passing through the adjunction. 
For example:
\vspace{-1ex}
{\small \[
    \inferrule* { \Gamma , \ottnt{A}   \shows  \Delta}
                { \Gamma , \oc  \ottnt{A}   \shows  \Delta}
\qquad \Rightarrow \qquad 
    \inferrule*{  \Gamma , \ottnt{A}   \shows  \Delta}
               {\inferrule*{  \Gamma ,  \bangfunctor{ \ottnt{A} }    \shows  \Delta}
                           {  \Gamma ,  F_\oc  \,  \bangfunctor{ \ottnt{A} }    \shows  \Delta}}
\qquad \qquad 
    \inferrule* { \Gamma ^{\oc}   \shows    \Delta ^{\wn}  , \ottnt{A} }
                { \Gamma ^{\oc}   \shows    \Delta ^{\wn}  , \oc  \ottnt{A} }
\qquad \Rightarrow \qquad
    \inferrule*{  \Gamma ^{\mode P}   \shows    \Delta ^{\mode C}  , \ottnt{A}  }
               {\inferrule*{  \Gamma ^{\mode P}   \Vdash    \Delta ^{\mode C}  ,  \bangfunctor{ \ottnt{A} }  }
                           {  \Gamma ^{\mode P}   \shows    \Delta ^{\mode C}  ,  F_\oc  \,  \bangfunctor{ \ottnt{A} }  }}
\]}

\begin{figure}[t]
\begin{subfigure}[b]{.6\textwidth}
\small
\[
\begin{array}{cc}
    \ottdrulepersXXPAxiom{$   \textsc{Ax}  _{  \mode{P}  }  ^{\Vdash} $}
    &\ottdrulepersXXCAxiom{$   \textsc{Ax}  _{  \mode{C}  }  ^{\Vdash} $}
    \\ \\
\multicolumn{2}{c}{
    \ottdrulepersXXMConjL{$    \mprod  _{  \mode{P}  }  ^{\Vdash}  \text{-L} $} \quad 
    \ottdrulepersXXMConjR{$    \mprod  _{  \mode{P}  }  ^{\Vdash}  \text{-R} $}
}
    \\ \\
    \ottdrulepersXXPOneL{$    1  _{  \mode{P}  }  ^{\Vdash}  \text{-L} $}
    &\ottdrulepersXXPOneR{$    1  _{  \mode{P}  }  ^{\Vdash}  \text{-R} $}
    \\ \\
\multicolumn{2}{c}{
    \ottdrulepersXXMDisjL{$    \msum  _{  \mode{C}  }  ^{\Vdash}  \text{-L} $} \quad 
    \ottdrulepersXXMDisjR{$    \msum  _{  \mode{C}  }  ^{\Vdash}  \text{-R} $}
}
    \\ \\
    \ottdrulepersXXCBotL{$    \bot  _{  \mode{C}  }  ^{\Vdash}  \text{-L} $}
    &\ottdrulepersXXCBotR{$    \bot  _{  \mode{C}  }  ^{\Vdash}  \text{-R} $}
\end{array}
\]
\vspace{-1em}
\caption{Inference Rules for Persistent Sequent}
\label{fig:linrules-persistent}
\end{subfigure}
\begin{subfigure}[b]{.4\textwidth}
\small
\[
\begin{array}{cc}
     \ottdrulelinXXFBangL{$   F_\oc   \text{-L} $}
    &\ottdrulelinXXFBangR{$   F_\oc   \text{-R} $}
   \\ \\
    \ottdrulelinXXFWhynotL{$   F_\wn   \text{-L} $}
    &\ottdrulelinXXFWhynotR{$   F_\wn   \text{-R} $}
    \\ \\
     \ottdrulelinXXBangL{$   \bangfunctor{ \ottsym{-} }   \text{-L} $}
    &\ottdrulepersXXBangR{$   \bangfunctor{ \ottsym{-} }   \text{-R} $}
   \\ \\
    \ottdrulepersXXWhynotL{$   \whynotfunctor{ \ottsym{-} }   \text{-L} $}
    &\ottdrulelinXXWhynotR{$   \whynotfunctor{ \ottsym{-} }   \text{-R} $}
\end{array}
\]
\vspace{-1em}
\caption{Adjunction Inference Rules}
\label{fig:structrules}
\end{subfigure}
\begin{subfigure}[b]{\textwidth}
\small
\[ \begin{array}{cc}
    \\
    \multicolumn{2}{c}{\ottdrulelinXXCutLinear{$   \textsc{Cut}  _{  \mode{L}  }  ^{\vdash} $}}
    \\ \\
    \ottdrulelinXXCutProducer{$   \textsc{Cut}  _{  \mode{P}  }  ^{\vdash} $}
    &
    \ottdrulepersXXCutProducer{$   \textsc{Cut}  _{  \mode{P}  }  ^{\Vdash} $}
    \\ \\
    \ottdrulelinXXCutConsumer{$   \textsc{Cut}  _{  \mode{C}  }  ^{\vdash} $}
    &
    \ottdrulepersXXCutConsumer{$   \textsc{Cut}  _{  \mode{C}  }  ^{\Vdash} $}
\end{array} \]
\vspace{-1em}
\caption{$ \textsc{Cut} $ Inference Rules}
\label{fig:cutrules}
\end{subfigure}
\caption{Persistent and Auxiliary Inference Rules}
\label{fig:auxrules}
\end{figure}
%

\vspace{-1.5em}
\paragraph*{Displacement.}

The commas on the left-hand-side of both the linear and persistent 
sequents intuitively correspond to the $ \mprod $ operator, 
and the commas on the right correspond to $ \msum $.
This correspondence motivates the context restriction in the rules that move
between the linear and persistent regimes. The restriction ensures that 
almost all of the propositions
have the ``natural'' mode---producers on the left and consumers on the right.
We say ``almost'' because the principal formula in each of these rules defies this
classification. We call such propositions \emph{displaced}.

\begin{definition}
    In a derivation of $\Gamma  \Vdash  \Delta$, a producer $\ottnt{P}$ is displaced 
    if it appears in $\Delta$. A consumer $\ottnt{C}$ is displaced if it appears in $\Gamma$.
\end{definition}

\begin{proposition}[Displacement]\label{prop:displacement}
    Every derivation of $\Gamma  \Vdash  \Delta$ contains exactly one
    displaced proposition.
\end{proposition}
\begin{proof}
    By induction on the derivation.
\end{proof}

\paragraph*{Cut.}

The cut rules are presented in \Figure{cutrules}. 
Notice that the rules with persistent cut terms have the following property:
whenever the cut term is displaced in a subderivation,
that derivation must be persistent and satisfy the restrictions of \Proposition{displacement}.
Intuitively, only persistent judgments can derive displaced propositions.

To show admissibility of the $ \textsc{Cut} $ rules, it is sufficient to show
admissibility of an equivalent set of rules called $  \textsc{Cut}  + $.
The versions differ in their treatment of persistent cut terms. The $  \textsc{Cut}  + $ formulation
uses the observation that when a persistent proposition is \emph{not} displaced in a sequent,
it can be replicated any number of times. Let $ ( \ottnt{X} )_{ n } $ be $n$ copies of
a proposition $\ottnt{X}$. It is easy to see that the
following propositions are admissible in the linear sequent (and similarly for the persistent sequent):
{\small \begin{align*} 
    \inferrule* { \Gamma ,  ( \ottnt{P} )_{ n }    \shows  \Delta} { \Gamma , \ottnt{P}   \shows  \Delta} 
    \qquad\qquad
    \inferrule* {\Gamma  \shows   \Delta ,  ( \ottnt{C} )_{ n }  } {\Gamma  \shows   \Delta , \ottnt{C} }
\end{align*}}
Thus the $  \textsc{Cut}  + $ rules, given 
below, are equivalent to the $ \textsc{Cut} $ rules.
{\small
 \[ \begin{array}{cc}
     \multicolumn{2}{c}{\ottdrulelinXXCutLinear{$    \textsc{Cut}  +  _{  \mode{L}  }  ^{\vdash} $}}
     \\ \\
     \ottdrulelinXXCutProducerPlus{$    \textsc{Cut}  +  _{  \mode{P}  }  ^{\vdash} $}
     &\qquad
     \ottdrulepersXXCutProducerPlus{$    \textsc{Cut}  +  _{  \mode{P}  }  ^{\Vdash} $}
     \\ \\
     \ottdrulelinXXCutConsumerPlus{$    \textsc{Cut}  +  _{  \mode{C}  }  ^{\vdash} $}
     &\qquad
     \ottdrulepersXXCutConsumerPlus{$    \textsc{Cut}  +  _{  \mode{C}  }  ^{\Vdash} $}
 \end{array} \]}

\begin{lemma}[$  \textsc{Cut}  + $ Admissibility]
    \label{lem:admissibility}
    The $  \textsc{Cut}  + $ rules are admissible in $ \textsc{LPC} $.
\end{lemma}
\begin{proof}
    Let $\mathcal{D}_{{\mathrm{1}}}$ and $\mathcal{D}_{{\mathrm{2}}}$ be the hypotheses of one of the cut rules.
    The proof is by induction on the cut term primarily
    and the sum of the depths of $\mathcal{D}_{{\mathrm{1}}}$ and $\mathcal{D}_{{\mathrm{2}}}$ secondly.

    \noindent
    \begin{minipage}[t]{0.70\textwidth}
        \qquad 1.\enspace Suppose $\mathcal{D}_{{\mathrm{1}}}$ or $\mathcal{D}_{{\mathrm{2}}}$ ends in a weakening or contraction rule on the cut term.
        In particular, consider the weakening case where the cut term is a producer
        and $\mathcal{D}_{{\mathrm{2}}}$ is a linear judgment. In this case $\mathcal{D}_{{\mathrm{1}}}$ is a derivation
        of $ \Gamma_{{\mathrm{1}}} ^{\mode P}   \Vdash    \Delta_{{\mathrm{1}}} ^{\mode C}  , \ottnt{P} $ and $\mathcal{D}_{{\mathrm{2}}}$ is the derivation shown to the right.
        By the inductive hypothesis on $\ottnt{P}$, $\mathcal{D}_{{\mathrm{1}}}$ and $\mathcal{D}'_{{\mathrm{2}}}$,
        there exists a cut-free derivation of \mbox{$  \Gamma_{{\mathrm{1}}} ^{\mode P}  , \Gamma_{{\mathrm{2}}}   \shows    \Delta_{{\mathrm{1}}} ^{\mode C}  , \Delta_{{\mathrm{2}}} $}.
    \end{minipage} \quad \begin{minipage}[t]{0.25\textwidth}
        \[
           \mathcal{D}_{{\mathrm{2}}} = \inferrule* [right=$  \textsc{W}  \text{-L} $]
             { \inferrule* { \mathcal{D}'_{{\mathrm{2}}} } {  \Gamma_{{\mathrm{2}}} ,  ( \ottnt{P} )_{ n }    \shows  \Delta_{{\mathrm{2}}} } }
             { \Gamma_{{\mathrm{2}}} ,  ( \ottnt{P} )_{ n  +   1  }    \shows  \Delta_{{\mathrm{2}}}}
        \]
    \end{minipage}

    \begin{asparaenum} \setcounter{enumi}{1}
    \item If $\mathcal{D}_{{\mathrm{1}}}$ or $\mathcal{D}_{{\mathrm{2}}}$ is an axiom, the case is trivial.

    \item Suppose the cut term is the principle formula in both $\mathcal{D}_{{\mathrm{1}}}$ and $\mathcal{D}_{{\mathrm{2}}}$
        (excluding weakening and contraction rules). We consider a few of the subcases here:
        \begin{asparaenum}
        \item[($  \mprod  _{  \mode{L}  } $)] 
                 \begin{tabular}{c}$\mathcal{D}_{{\mathrm{1}}} = \inferrule* [right=$    \mprod  _{  \mode{L}  }  ^{\vdash}  \text{-R} $]
                    { \inferrule* { \mathcal{D}_{{\mathrm{11}}} } { \Gamma_{{\mathrm{11}}}  \shows   \Delta_{{\mathrm{11}}} , \ottnt{A_{{\mathrm{1}}}}  }  \\
                      \inferrule* { \mathcal{D}_{{\mathrm{12}}} } { \Gamma_{{\mathrm{12}}}  \shows   \Delta_{{\mathrm{12}}} , \ottnt{A_{{\mathrm{2}}}}  } }
                    { \Gamma_{{\mathrm{11}}} , \Gamma_{{\mathrm{12}}}   \shows    \Delta_{{\mathrm{11}}} , \Delta_{{\mathrm{12}}}  , \ottnt{A_{{\mathrm{1}}}}  \mprod  \ottnt{A_{{\mathrm{2}}}} }
                \qquad\text{and}\qquad
                   \mathcal{D}_{{\mathrm{2}}} = \inferrule* [right=$    \mprod  _{  \mode{L}  }  ^{\vdash}  \text{-L} $]
                    { \inferrule* { \mathcal{D}'_{{\mathrm{2}}} } {   \Gamma_{{\mathrm{2}}} , \ottnt{A_{{\mathrm{1}}}}  , \ottnt{A_{{\mathrm{2}}}}   \shows  \Delta_{{\mathrm{2}}} } }
                    { \Gamma_{{\mathrm{2}}} , \ottnt{A_{{\mathrm{1}}}}  \mprod  \ottnt{A_{{\mathrm{2}}}}   \shows  \Delta_{{\mathrm{2}}}}$
                \end{tabular}
                \\~\\
                \noindent
                By the inductive hypothesis on $\ottnt{A_{{\mathrm{2}}}}$, $\mathcal{D}_{{\mathrm{12}}}$ and $\mathcal{D}'_{{\mathrm{2}}}$,
                there exists a derivation $\mathcal{E}$ of $  \Gamma_{{\mathrm{12}}} , \Gamma_{{\mathrm{2}}}  , \ottnt{A_{{\mathrm{1}}}}   \shows   \Delta_{{\mathrm{12}}} , \Delta_{{\mathrm{2}}} $.
                Then the desired derivation of $  \Gamma_{{\mathrm{11}}} , \Gamma_{{\mathrm{12}}}  , \Gamma_{{\mathrm{2}}}   \shows    \Delta_{{\mathrm{11}}} , \Delta_{{\mathrm{12}}}  , \Delta_{{\mathrm{2}}} $ exists
                by the inductive hypothesis on $\ottnt{A_{{\mathrm{1}}}}$, $\mathcal{D}_{{\mathrm{11}}}$ and $\mathcal{E}$.

        \item[($  \mprod  _{  \mode{P}  } $)]
                \begin{tabular}{c}$\mathcal{D}_{{\mathrm{1}}} = \inferrule* [right=$    \mprod  _{  \mode{P}  }  ^{\Vdash}  \text{-R} $]
                    { \inferrule* { \mathcal{D}_{{\mathrm{11}}} } {  \Gamma_{{\mathrm{11}}} ^{\mode P}    \Vdash     \Delta_{{\mathrm{11}}} ^{\mode C}  , \ottnt{P_{{\mathrm{1}}}}  }  \\
                      \inferrule* { \mathcal{D}_{{\mathrm{12}}} } {  \Gamma_{{\mathrm{12}}} ^{\mode P}    \Vdash     \Delta_{{\mathrm{12}}} ^{\mode C}  , \ottnt{P_{{\mathrm{2}}}}  } }
                    {   \Gamma_{{\mathrm{11}}} ^{\mode P}  , \Gamma_{{\mathrm{12}}}  ^{\mode P}   \Vdash      \Delta_{{\mathrm{11}}} ^{\mode C}  , \Delta_{{\mathrm{12}}}  ^{\mode C}  , \ottnt{P_{{\mathrm{1}}}}  \mprod  \ottnt{P_{{\mathrm{2}}}} }
                \quad\text{and}\quad
                   \mathcal{D}_{{\mathrm{2}}} = \inferrule* [right=$    \mprod  _{  \mode{P}  }  ^{\vdash}  \text{-L} $]
                    { \inferrule* { \mathcal{D}'_{{\mathrm{2}}} } {    \Gamma_{{\mathrm{2}}} ,  ( \ottnt{P_{{\mathrm{1}}}}  \mprod  \ottnt{P_{{\mathrm{2}}}} )_{ n }   , \ottnt{P_{{\mathrm{1}}}}  , \ottnt{P_{{\mathrm{2}}}}   \shows  \Delta_{{\mathrm{2}}} } }
                    { \Gamma_{{\mathrm{2}}} ,  ( \ottnt{P_{{\mathrm{1}}}}  \mprod  \ottnt{P_{{\mathrm{2}}}} )_{ n  +   1  }    \shows  \Delta_{{\mathrm{2}}}}$
                \end{tabular}
                \\~\\
                \noindent
                The inductive hypothesis on $\ottnt{P_{{\mathrm{1}}}}  \mprod  \ottnt{P_{{\mathrm{2}}}}$, $\mathcal{D}_{{\mathrm{1}}}$ itself
                and $\mathcal{D}'_{{\mathrm{2}}}$ gives us a derivation $\mathcal{E}$ of
                \[       \Gamma_{{\mathrm{11}}} ^{\mode P}  , \Gamma_{{\mathrm{12}}}  ^{\mode P}  , \Gamma_{{\mathrm{2}}}  , \ottnt{P_{{\mathrm{1}}}}  , \ottnt{P_{{\mathrm{2}}}}   \shows      \Delta_{{\mathrm{11}}} ^{\mode C}  , \Delta_{{\mathrm{12}}}  ^{\mode C}  , \Delta_{{\mathrm{2}}} . \]
                Multiple applications of the inductive hypothesis give the following
                derivation:
                \[
                    \inferrule* [Right=IH($\ottnt{P_{{\mathrm{1}}}}$)]
                      { \inferrule* { \mathcal{D}_{{\mathrm{11}}} } {  \Gamma_{{\mathrm{11}}} ^{\mode P}    \Vdash     \Delta_{{\mathrm{11}}} ^{\mode C}  , \ottnt{P_{{\mathrm{1}}}}  } 
                      \\
                        \inferrule* [Right=IH($\ottnt{P_{{\mathrm{2}}}}$)]
                        { \inferrule* { \mathcal{D}_{{\mathrm{12}}} } {  \Gamma_{{\mathrm{12}}} ^{\mode P}    \Vdash     \Delta_{{\mathrm{12}}} ^{\mode C}  , \ottnt{P_{{\mathrm{2}}}}  } 
                        \\
                          \inferrule* { \mathcal{E} } {       \Gamma_{{\mathrm{11}}} ^{\mode P}  , \Gamma_{{\mathrm{12}}}  ^{\mode P}  , \Gamma_{{\mathrm{2}}}  , \ottnt{P_{{\mathrm{1}}}}  , \ottnt{P_{{\mathrm{2}}}}   \shows      \Delta_{{\mathrm{11}}} ^{\mode C}  , \Delta_{{\mathrm{12}}}  ^{\mode C}  , \Delta_{{\mathrm{2}}}  } 
                        }
                        {       \Gamma_{{\mathrm{12}}} ^{\mode P}  , \Gamma_{{\mathrm{11}}}  ^{\mode P}  , \Gamma_{{\mathrm{12}}}  ^{\mode P}  , \Gamma_{{\mathrm{2}}}  , \ottnt{P_{{\mathrm{1}}}}   \shows        \Delta_{{\mathrm{12}}} ^{\mode C}  , \Delta_{{\mathrm{11}}}  ^{\mode C}  , \Delta_{{\mathrm{12}}}  ^{\mode C}  , \Delta_{{\mathrm{2}}} }
                      }
                      {        \Gamma_{{\mathrm{11}}} ^{\mode P}  , \Gamma_{{\mathrm{12}}}  ^{\mode P}  , \Gamma_{{\mathrm{11}}}  ^{\mode P}  , \Gamma_{{\mathrm{12}}}  ^{\mode P}  , \Gamma_{{\mathrm{2}}}   \shows          \Delta_{{\mathrm{11}}} ^{\mode C}  , \Delta_{{\mathrm{12}}}  ^{\mode C}  , \Delta_{{\mathrm{11}}}  ^{\mode C}  , \Delta_{{\mathrm{12}}}  ^{\mode C}  , \Delta_{{\mathrm{2}}} }
                \]
                Because the replicated contexts are made up exclusively of non-displaced
                propositions, it is possible to apply contraction multiple times
                to obtain the desired sequent.

        \item[($ F_\oc $)]
                \begin{tabular}{c}$ \mathcal{D}_{{\mathrm{1}}} = \inferrule* [right=$   F_\oc   \text{-R} $]
                    { \inferrule* { \mathcal{D}'_{{\mathrm{1}}} } {  \Gamma_{{\mathrm{1}}} ^{\mode P}    \Vdash     \Delta_{{\mathrm{1}}} ^{\mode C}  , \ottnt{P}  } }
                    { \Gamma_{{\mathrm{1}}} ^{\mode P}   \shows    \Delta_{{\mathrm{1}}} ^{\mode C}  ,  F_\oc  \, \ottnt{P} }
                \qquad\text{and}\qquad
                   \mathcal{D}_{{\mathrm{2}}} = \inferrule* [right=$   F_\oc   \text{-L} $]
                    { \inferrule* { \mathcal{D}'_{{\mathrm{2}}} } {  \Gamma_{{\mathrm{2}}} , \ottnt{P}   \shows  \Delta } }
                    { \Gamma_{{\mathrm{2}}} ,  F_\oc  \, \ottnt{P}   \shows  \Delta}
                $\end{tabular}
                \\~\\ \noindent
                Because $\mathcal{D}'_{{\mathrm{1}}}$ is a persistent derivation,
                we can apply the inductive hypothesis for $\ottnt{P}$ 
                with $n=1$ to obtain the desired derivation.
        \end{asparaenum}

    \item Suppose the cut term is \emph{not} the principle formula in $\mathcal{D}_{{\mathrm{1}}}$ or $\mathcal{D}_{{\mathrm{2}}}$.
        Most of the subcases are straightforward in that the last rule in the derivation
        commutes with the inductive hypotheses. 

        If the cut term is a producer, then $\mathcal{D}_{{\mathrm{1}}}$ is a persistent
            judgment so it cannot be the case that the last rule
            of $\mathcal{D}_{{\mathrm{1}}}$ is an $ F_\oc $ rule or a $   \bangfunctor{ \ottsym{-} }   \text{-L} $ or $   \whynotfunctor{ \ottsym{-} }   \text{-R} $
            rule. But it also cannot be the case that the last rule
            in $\mathcal{D}_{{\mathrm{1}}}$ is a $   \bangfunctor{ \ottsym{-} }   \text{-R} $ or $   \whynotfunctor{ \ottsym{-} }   \text{-L} $ rule
            because there is a non-principle formula---namely, the cut formula---which
            is in a displaced position in the derivation.

        \begin{minipage}[t]{0.65\textwidth}
        Suppose on the other hand that the cut term is a consumer and $\mathcal{D}_{{\mathrm{1}}}$
        is the derivation to the right. Then $\mathcal{D}_{{\mathrm{2}}}$ is a derivation of
        $  \Gamma_{{\mathrm{2}}} ^{\mode P}  , \ottnt{C}   \Vdash   \Delta_{{\mathrm{2}}} ^{\mode C} $.
        By the inductive hypothesis on $\ottnt{C}$, $\mathcal{D}'_{{\mathrm{1}}}$ and $\mathcal{D}_{{\mathrm{2}}}$, there is
        a derivation $\mathcal{E}$ of $   \Gamma_{{\mathrm{1}}} ^{\mode P}  , \Gamma_{{\mathrm{2}}}  ^{\mode P}   \shows      \Delta_{{\mathrm{1}}} ^{\mode C}  , \ottnt{A}  , \Delta_{{\mathrm{2}}}  ^{\mode C} $. Because the contexts in
        $\mathcal{D}_{{\mathrm{2}}}$ were undisplaced, it is possible to apply the $   \bangfunctor{ \ottsym{-} }   \text{-R} $ rule to
        $\mathcal{E}$ to obtain a derivation of $   \Gamma_{{\mathrm{1}}} ^{\mode P}  , \Gamma_{{\mathrm{2}}}  ^{\mode P}   \Vdash      \Delta_{{\mathrm{1}}} ^{\mode C}  ,  \bangfunctor{ \ottnt{A} }   , \Delta_{{\mathrm{2}}}  ^{\mode C} .$
        \end{minipage} \quad \begin{minipage}[t]{0.30\textwidth}
            \[ \mathcal{D}_{{\mathrm{1}}} = \inferrule* [right=$   \bangfunctor{ \ottsym{-} }   \text{-R} $]
                { \inferrule* { \mathcal{D}'_{{\mathrm{1}}} } {  \Gamma_{{\mathrm{1}}} ^{\mode P}   \shows     \Delta_{{\mathrm{1}}} ^{\mode C}  , \ottnt{A}  ,  ( \ottnt{C} )_{ n }   } }
                { \Gamma_{{\mathrm{1}}} ^{\mode P}   \Vdash     \Delta_{{\mathrm{1}}} ^{\mode C}  ,  \bangfunctor{ \ottnt{A} }   ,  ( \ottnt{C} )_{ n }  }
            \]
        \end{minipage}
    \end{asparaenum}
For the full proof of \Lemma{admissibility},
see the accompanying technical report~\cite{PaykinZdancewic14tech}.
\end{proof}

\begin{theorem}[$ \textsc{Cut} $ Admissibility]
    \label{thm:elimination}
    The $ \textsc{Cut} $ rules in \Figure{cutrules} are admissible in $ \textsc{LPC} $.
\end{theorem}

\paragraph*{Duality.}

Every rule in the $ \textsc{LPC} $ inference rules has a clear dual, but 
unlike standard presentations of classical linear logic,
$ \textsc{LPC} $ does not contain an explicit duality operator $ \dualize{ \ottsym{(}  \ottsym{-}  \ottsym{)} } $,
nor a linear implication $ \lolto $ with which to encode duality. Instead, 
we define $ \dualize{ \ottsym{(}  \ottsym{-}  \ottsym{)} } $ to be a meta-operation on propositions
and prove that the following duality rules are admissible in $ \textsc{LPC} $:
{\small \[ 
    \ottdrulelinXXLinDualL{$  \dualize{(-)}  \text{-L} $}
    \qquad\qquad
    \ottdrulelinXXLinDualR{$  \dualize{(-)}  \text{-R} $}
\] }
In fact, there are three versions of this duality operation: $ \dualize{ \ottsym{(}  \ottsym{-}  \ottsym{)} } $ for linear propositions,
$ \star{ \ottsym{(}  \ottsym{-}  \ottsym{)} } $ for producers and $ \lowerstar{ \ottsym{(}  \ottsym{-}  \ottsym{)} } $ for consumers. 
For a linear
proposition $\ottnt{A}$, $ \dualize{ \ottnt{A} } $ is linear, but for a producer $\ottnt{P}$, $ \star{ \ottnt{P} } $ is a
consumer, and for a consumer $\ottnt{C}$, $ \lowerstar{ { \ottnt{C} } } $ is a producer. 
We define these (invertible) duality operations as follows:
\[ \begin{aligned}
     \dualize{  \top  }     &:= \ottsym{0} \\
     \dualize{ \ottsym{0} }       &:=  \top  \\
     \dualize{ \ottsym{(}  \ottnt{A}  \aprod  \ottnt{B}  \ottsym{)} }  &:=  \dualize{  \dualize{ \ottnt{A} }   \asum  \ottnt{B} }  \\
     \dualize{ \ottsym{(}  \ottnt{A}  \asum  \ottnt{B}  \ottsym{)} }  &:=  \dualize{  \dualize{ \ottnt{A} }   \aprod  \ottnt{B} } 
\end{aligned} \qquad\begin{aligned}
     \dualize{  1_{\mode L}  }       &:=  \bot_{\mode L}  \\
     \dualize{  \bot_{\mode L}  }     &:=  1_{\mode L}  \\
     \dualize{ \ottsym{(}  \ottnt{A}  \mprod  \ottnt{B}  \ottsym{)} }  &:=  \dualize{  \dualize{ \ottnt{A} }   \msum  \ottnt{B} }  \\
     \dualize{ \ottsym{(}  \ottnt{A}  \msum  \ottnt{B}  \ottsym{)} }  &:=  \dualize{  \dualize{ \ottnt{A} }   \mprod  \ottnt{B} }  
\end{aligned} \qquad \begin{aligned}
     \star{  1_{\mode P}  }    &:=  \bot_{\mode C}  \\
     \lowerstar{ {  \bot_{\mode C}  } }  &:=  1_{\mode P}  \\
     \star{ \ottsym{(}  \ottnt{P}  \mprod  \ottnt{Q}  \ottsym{)} } &:=  \star{ \ottnt{P} }   \msum   \star{ \ottnt{Q} }  \\
     \lowerstar{ { \ottsym{(}  \ottnt{C}  \msum  D  \ottsym{)} } } &:=  \lowerstar{ { \ottnt{C} } }   \mprod   \lowerstar{ { D } } 
\end{aligned} \qquad\begin{aligned}
     \dualize{ \ottsym{(}   F_\oc  \, \ottnt{P}  \ottsym{)} }  &:=  F_\wn  \,  \star{ \ottnt{P} }  \\
     \dualize{ \ottsym{(}   F_\wn  \, \ottnt{C}  \ottsym{)} }  &:=  F_\oc  \,  \lowerstar{ { \ottnt{C} } }  \\
     \star{  \bangfunctor{ \ottnt{A} }  }   &:=  \whynotfunctor{  \dualize{ \ottnt{A} }  }  \\
     \lowerstar{ {  \whynotfunctor{ \ottnt{A} }  } }   &:=  \bangfunctor{  \dualize{ \ottnt{A} }  } 
\end{aligned} \]

\begin{wrapfigure}{r}{0.45\textwidth}
\[ \begin{array}{cc}
    \multicolumn{2}{c}{\ottdrulelinXXLinDualL{$  \dualize{(-)}  \text{-L} $}}
    \\ \\
    \ottdrulelinXXProdDualL{$   \star{(-)}  ^{\vdash}  \text{-L} $}
    &
    \ottdrulelinXXConsDualL{$   \lowerstar{(-)}  ^{\vdash}  \text{-L} $}
    \\ \\
    \ottdrulepersXXProdDualL{$   \star{(-)}  ^{\Vdash}  \text{-L} $}
    &
    \ottdrulepersXXConsDualL{$   \lowerstar{(-)}  ^{\Vdash}  \text{-L} $}
\end{array} \]
\caption{Left Duality Inference Rules}
\label{fig:dual}
\end{wrapfigure}

We will show that the inference rules given in \Figure{dual}
(as well as the respective right rules) are admissible in $ \textsc{LPC} $.

\begin{lemma} The following axioms hold in $ \textsc{LPC} $:\footnotemark
{\small     \[ 
        \inferrule* { } { \ottnt{A} ,  \dualize{ \ottnt{A} }    \shows   \cdot }
        \qquad
        \inferrule* { } { \cdot   \shows   \ottnt{A} ,  \dualize{ \ottnt{A} }  }
        \qquad
        \inferrule* { } { \ottnt{P} ,  \star{ \ottnt{P} }    \Vdash   \cdot }
        \qquad
        \inferrule* { } { \cdot   \Vdash   \ottnt{P} ,  \star{ \ottnt{P} }  }
    \]}
\end{lemma}
\begin{proof}
By mutual induction on the proposition. 
\end{proof}
\footnotetext{Notice that the consumer case $ \ottnt{C} ,  \lowerstar{ { \ottnt{C} } }    \Vdash   \cdot $ is encompassed
by the producer case where $\ottnt{P}= \lowerstar{ { \ottnt{C} } } $.}

The variations $ \ottnt{P} ,  \star{ \ottnt{P} }    \shows   \cdot $ and $ \cdot   \shows   \ottnt{P} ,  \star{ \ottnt{P} }  $ on the other hand
cannot be proved by induction because
of the subcase $\ottnt{P}= \bangfunctor{ \ottnt{A} } $; there is no way to apply the inductive hypothesis 
to the goal $  \bangfunctor{ \ottnt{A} }  ,  \whynotfunctor{  \dualize{ \ottnt{A} }  }    \shows   \cdot $. However we can construct the desired derivations
using cut rules:
\[
    \inferrule* 
    {\inferrule*{ }{ \star{ \ottnt{P} }   \shows   \star{ \ottnt{P} } }
    \\
     \inferrule*{ }{ \ottnt{P} ,  \star{ \ottnt{P} }    \Vdash   \cdot }
    }
    { \ottnt{P} ,  \star{ \ottnt{P} }    \shows   \cdot }
    \qquad\qquad
    \inferrule*
    {\inferrule*{ }{ \cdot   \Vdash   \ottnt{P} ,  \star{ \ottnt{P} }  }
    \\
     \inferrule*{ }{\ottnt{P}  \shows  \ottnt{P}}
    }
    { \cdot   \shows   \ottnt{P} ,  \star{ \ottnt{P} }  }
\]

\begin{theorem}
    The duality rules in \Figure{dual} (and thus the corresponding right rules) 
    are admissible in $ \textsc{LPC} $.
\end{theorem}
\begin{proof}
Three of the rules can be generated by a straightforward application of cut:
{\small    \[ 
        \inferrule* [right=$   \textsc{Cut}  _{  \mode{L}  }  ^{\vdash} $]
        {\Gamma  \shows   \Delta , \ottnt{A}  
        \\
         \inferrule* { } { \ottnt{A} ,  \dualize{ \ottnt{A} }    \shows   \cdot }
        }
        { \Gamma ,  \dualize{ \ottnt{A} }    \shows  \Delta}
        \qquad\qquad
        \inferrule* [right=$   \textsc{Cut}  _{  \mode{C}  }  ^{\vdash} $]
        {\Gamma  \shows   \Delta , \ottnt{C} 
        \\
        \inferrule*{ }{ \ottnt{C} ,  \lowerstar{ { \ottnt{C} } }    \Vdash   \cdot }
        }
        { \Gamma ,  \lowerstar{ { \ottnt{C} } }    \shows  \Delta}
        \qquad\qquad
        \inferrule* [right=$   \textsc{Cut}  _{  \mode{C}  }  ^{\Vdash} $]
        {\Gamma  \Vdash   \Delta , \ottnt{C} 
        \\
        \inferrule*{ }{ \ottnt{C} ,  \lowerstar{ { \ottnt{C} } }    \Vdash   \cdot }
        }
        { \Gamma ,  \lowerstar{ { \ottnt{C} } }    \Vdash  \Delta}
    \] }
When we try to do the same for the left producer rules,
the context restriction around the displaced cut term
leads to the following derivations:
\[
        \inferrule* [right=$   \textsc{Cut}  _{  \mode{P}  }  ^{\Vdash} $]
        { \Gamma ^{\mode P}   \Vdash    \Delta ^{\mode C}  , \ottnt{P} 
        \\
         \inferrule* { } { \ottnt{P} ,  \star{ \ottnt{P} }    \Vdash   \cdot }
        }
        {  \Gamma ^{\mode P}  ,  \star{ \ottnt{P} }    \Vdash   \Delta ^{\mode C} }
        \qquad\qquad
        \inferrule* [right=$   \textsc{Cut}  _{  \mode{P}  }  ^{\vdash} $]
        { \Gamma ^{\mode P}   \Vdash    \Delta ^{\mode C}  , \ottnt{P}  
        \\
        \inferrule* { } { \ottnt{P} ,  \star{ \ottnt{P} }    \shows   \cdot }
        }
        {  \Gamma ^{\mode P}  ,  \star{ \ottnt{P} }    \shows   \Delta ^{\mode C} } 
\]
For the first of these, recall that due to displacement, every derivation of $ \Gamma ,  \star{ \ottnt{P} }    \Vdash  \Delta$
in fact has the restriction that $\Gamma= \Gamma ^{\mode P} $ and $\Delta= \Delta ^{\mode C} $. So this derivation is actually
equivalent to the one in \Figure{dual}.
The second derivation, on the other hand, is not equivalent to the one in \Figure{dual},
nor an acceptable variant. The hypothesis and conclusion of the derivation
are different kinds of sequents, and linear propositions are completely 
excluded from the contexts.

Instead we can prove the more general form of the rule directly:
For any derivation $\mathcal{D}$ of $\Gamma  \shows   \Delta , \ottnt{P} $, there is a derivation of
$ \Gamma ,  \star{ \ottnt{P} }    \shows  \Delta$. We prove this by induction on $\mathcal{D}$. Most of the cases 
commute directly with the inductive hypothesis, which the following exception:
If $\mathcal{D}$ is the axiom $\ottnt{P}  \shows  \ottnt{P}$ then there
is a derivation of $ \ottnt{P} ,  \star{ \ottnt{P} }    \shows   \cdot $, as expected.
\end{proof}

\paragraph*{Consistency.}

Define the negation of a linear proposition to be $ \neg \ottnt{A} := \dualize{ \ottnt{A} }   \msum  \ottsym{0}$.

\begin{theorem}[Consistency]
    There is no proposition $\ottnt{A}$ such that $\ottnt{A}$ and $ \neg \ottnt{A} $
    are both provable in $ \textsc{LPC} $.
\end{theorem}
\noindent
\begin{minipage}[t]{0.70\textwidth}
\begin{proof}
     Suppose there were such an $\ottnt{A}$, along with derivations $\mathcal{D}_{{\mathrm{1}}}$
     of $ \cdot   \shows  \ottnt{A}$ and $\mathcal{D}_{{\mathrm{2}}}$ of $ \cdot   \shows   \dualize{ \ottnt{A} }   \msum  \ottsym{0}$. 
    Then there exists a derivation of $ \cdot   \shows  \ottsym{0}$ as seen 
    to the right.
     However, there is no cut-free proof of $ \cdot   \shows  \ottsym{0}$
     in $ \textsc{LPC} $, which contradicts cut admissibility.
 \end{proof}
\end{minipage} \quad
\begin{minipage}[t]{0.30\textwidth}
     \vspace{-7mm}
{\small     \[ 
         \inferrule* [right=$   \textsc{Cut}  _{  \mode{L}  }  ^{\vdash} $]
         { \inferrule* { \mathcal{D}_{{\mathrm{2}}} } {  \cdot   \shows    \dualize{ \ottnt{A} }  , \ottsym{0}  }  
         \\
         \inferrule* [Right=$  \dualize{(-)}  \text{-L} $]
             { \inferrule* { \mathcal{D}_{{\mathrm{1}}} } {  \cdot   \shows  \ottnt{A} } }
             { \dualize{ \ottnt{A} }   \shows   \cdot }
         }
         { \cdot   \shows  \ottsym{0}}
     \]}
\end{minipage}

\section{Categorical Model}
\label{sec:category}
In this section we describe a categorical axiomatization of $ \textsc{LPC} $ based on
the three-category \Figure{3cat}. 
Certain definitions have been omitted for brevity; these can be found
in the companion paper \cite{PaykinZdancewic14tech}.

\paragraph{Preliminaries.}
We start with some basic definitions about symmetric monoidal structures.

\begin{definition}\label{def:SMC}
    A \emph{symmetric monoidal category} is a category $\cat{C}$ equipped with 
    a bifunctor $ \mprod $, an object $ 1 $, and the following natural isomorphisms:
    \begin{align*}
    \begin{aligned}
          \alpha  _{  \ottnt{A_{{\mathrm{1}}}}  ,   \ottnt{A_{{\mathrm{2}}}}  ,  \ottnt{A_{{\mathrm{3}}}}   } &: \ottnt{A_{{\mathrm{1}}}}  \mprod  \ottsym{(}  \ottnt{A_{{\mathrm{2}}}}  \mprod  \ottnt{A_{{\mathrm{3}}}}  \ottsym{)} \rightarrow \ottsym{(}  \ottnt{A_{{\mathrm{1}}}}  \mprod  \ottnt{A_{{\mathrm{2}}}}  \ottsym{)}  \mprod  \ottnt{A_{{\mathrm{3}}}} \\
          \sigma  _{  \ottnt{A}  ,  \ottnt{B}  } &: \ottnt{A}  \mprod  \ottnt{B} \rightarrow \ottnt{B}  \mprod  \ottnt{A}
    \end{aligned} \qquad\qquad
    \begin{aligned}
          \lambda  _{ \ottnt{A} } &:  1   \mprod  \ottnt{A} \rightarrow \ottnt{A} \\
          \rho  _{ \ottnt{A} } &: \ottnt{A}  \mprod   1  \rightarrow \ottnt{A} 
    \end{aligned}
    \end{align*}
    These must satisfy the following coherence conditions:
    \begin{align}
       &\ottnt{A_{{\mathrm{1}}}}  \mprod  \ottsym{(}  \ottnt{A_{{\mathrm{2}}}}  \mprod  \ottsym{(}  \ottnt{A_{{\mathrm{3}}}}  \mprod  \ottnt{A_{{\mathrm{4}}}}  \ottsym{)}  \ottsym{)}
        \xrightarrow{  \alpha  _{  \ottnt{A_{{\mathrm{1}}}}  \mprod  \ottnt{A_{{\mathrm{2}}}}  ,   \ottnt{A_{{\mathrm{3}}}}  ,  \ottnt{A_{{\mathrm{4}}}}   } \circ   \alpha  _{  \ottnt{A_{{\mathrm{1}}}}  ,   \ottnt{A_{{\mathrm{2}}}}  ,  \ottnt{A_{{\mathrm{3}}}}  \mprod  \ottnt{A_{{\mathrm{4}}}}   } }
        \ottsym{(}  \ottsym{(}  \ottnt{A_{{\mathrm{1}}}}  \mprod  \ottnt{A_{{\mathrm{2}}}}  \ottsym{)}  \mprod  \ottnt{A_{{\mathrm{3}}}}  \ottsym{)}  \mprod  \ottnt{A_{{\mathrm{4}}}}
        \xrightarrow{ \inv{  {   \alpha  _{  \ottnt{A_{{\mathrm{1}}}}  ,   \ottnt{A_{{\mathrm{2}}}}  ,  \ottnt{A_{{\mathrm{3}}}}   }  }  }   \mprod    \textrm{id}  _{ \ottnt{A_{{\mathrm{4}}}} } }
        \ottsym{(}  \ottnt{A_{{\mathrm{1}}}}  \mprod  \ottsym{(}  \ottnt{A_{{\mathrm{2}}}}  \mprod  \ottnt{A_{{\mathrm{3}}}}  \ottsym{)}  \ottsym{)}  \mprod  \ottnt{A_{{\mathrm{4}}}} \notag\\
     =~&\ottnt{A_{{\mathrm{1}}}}  \mprod  \ottsym{(}  \ottnt{A_{{\mathrm{2}}}}  \mprod  \ottsym{(}  \ottnt{A_{{\mathrm{3}}}}  \mprod  \ottnt{A_{{\mathrm{4}}}}  \ottsym{)}  \ottsym{)}
        \xrightarrow{  \textrm{id}  _{ \ottnt{A_{{\mathrm{1}}}} }   \mprod    \alpha  _{  \ottnt{A_{{\mathrm{2}}}}  ,   \ottnt{A_{{\mathrm{3}}}}  ,  \ottnt{A_{{\mathrm{4}}}}   } }
        \ottnt{A_{{\mathrm{1}}}}  \mprod  \ottsym{(}  \ottsym{(}  \ottnt{A_{{\mathrm{2}}}}  \mprod  \ottnt{A_{{\mathrm{3}}}}  \ottsym{)}  \mprod  \ottnt{A_{{\mathrm{4}}}}  \ottsym{)}
        \xrightarrow{  \alpha  _{  \ottnt{A_{{\mathrm{1}}}}  ,   \ottnt{A_{{\mathrm{2}}}}  \mprod  \ottnt{A_{{\mathrm{3}}}}  ,  \ottnt{A_{{\mathrm{4}}}}   } }
        \ottsym{(}  \ottnt{A_{{\mathrm{1}}}}  \mprod  \ottsym{(}  \ottnt{A_{{\mathrm{2}}}}  \mprod  \ottnt{A_{{\mathrm{3}}}}  \ottsym{)}  \ottsym{)}  \mprod  \ottnt{A_{{\mathrm{4}}}} 
    \end{align}
    \begin{align}
          \textrm{id}  _{ \ottnt{A} }   \mprod    \lambda  _{ \ottnt{B} }  =
        \ottnt{A}  \mprod  \ottsym{(}   1   \mprod  \ottnt{B}  \ottsym{)}
        \xrightarrow{  \alpha  _{  \ottnt{A}  ,    1   ,  \ottnt{B}   } }
        \ottsym{(}  \ottnt{A}  \mprod   1   \ottsym{)}  \mprod  \ottnt{B}
        \xrightarrow{  \rho  _{ \ottnt{A} }   \mprod    \textrm{id}  _{ \ottnt{B} } }
        \ottnt{A}  \mprod  \ottnt{B}
    \end{align} 
    \begin{align}
       &\ottnt{A_{{\mathrm{1}}}}  \mprod  \ottsym{(}  \ottnt{A_{{\mathrm{2}}}}  \mprod  \ottnt{A_{{\mathrm{3}}}}  \ottsym{)}
        \xrightarrow{  \textrm{id}  _{ \ottnt{A_{{\mathrm{1}}}} }   \mprod    \sigma  _{  \ottnt{A_{{\mathrm{2}}}}  ,  \ottnt{A_{{\mathrm{3}}}}  } }
        \ottnt{A_{{\mathrm{1}}}}  \mprod  \ottsym{(}  \ottnt{A_{{\mathrm{3}}}}  \mprod  \ottnt{A_{{\mathrm{2}}}}  \ottsym{)}
        \xrightarrow{  \alpha  _{  \ottnt{A_{{\mathrm{1}}}}  ,   \ottnt{A_{{\mathrm{3}}}}  ,  \ottnt{A_{{\mathrm{2}}}}   } }
        \ottsym{(}  \ottnt{A_{{\mathrm{1}}}}  \mprod  \ottnt{A_{{\mathrm{3}}}}  \ottsym{)}  \mprod  \ottnt{A_{{\mathrm{2}}}}
        \xrightarrow{  \sigma  _{  \ottnt{A_{{\mathrm{1}}}}  ,  \ottnt{A_{{\mathrm{3}}}}  }   \mprod    \textrm{id}  _{ \ottnt{A_{{\mathrm{2}}}} } }
        \ottsym{(}  \ottnt{A_{{\mathrm{3}}}}  \mprod  \ottnt{A_{{\mathrm{1}}}}  \ottsym{)}  \mprod  \ottnt{A_{{\mathrm{2}}}} \notag\\
     =~&\ottnt{A_{{\mathrm{1}}}}  \mprod  \ottsym{(}  \ottnt{A_{{\mathrm{2}}}}  \mprod  \ottnt{A_{{\mathrm{3}}}}  \ottsym{)}
        \xrightarrow{  \alpha  _{  \ottnt{A_{{\mathrm{1}}}}  ,   \ottnt{A_{{\mathrm{2}}}}  ,  \ottnt{A_{{\mathrm{3}}}}   } }
        \ottsym{(}  \ottnt{A_{{\mathrm{1}}}}  \mprod  \ottnt{A_{{\mathrm{2}}}}  \ottsym{)}  \mprod  \ottnt{A_{{\mathrm{3}}}}
        \xrightarrow{  \sigma  _{  \ottnt{A_{{\mathrm{1}}}}  \mprod  \ottnt{A_{{\mathrm{2}}}}  ,  \ottnt{A_{{\mathrm{3}}}}  } }
        \ottnt{A_{{\mathrm{3}}}}  \mprod  \ottsym{(}  \ottnt{A_{{\mathrm{1}}}}  \mprod  \ottnt{A_{{\mathrm{2}}}}  \ottsym{)}
        \xrightarrow{  \alpha  _{  \ottnt{A_{{\mathrm{3}}}}  ,   \ottnt{A_{{\mathrm{1}}}}  ,  \ottnt{A_{{\mathrm{2}}}}   } }
        \ottsym{(}  \ottnt{A_{{\mathrm{3}}}}  \mprod  \ottnt{A_{{\mathrm{1}}}}  \ottsym{)}  \mprod  \ottnt{A_{{\mathrm{2}}}}
    \end{align}
    \begin{align}
          \textrm{id}  _{ \ottnt{A}  \mprod  \ottnt{B} }  =
        \ottnt{A}  \mprod  \ottnt{B}
        \xrightarrow{  \sigma  _{  \ottnt{A}  ,  \ottnt{B}  } }
        \ottnt{B}  \mprod  \ottnt{A}
        \xrightarrow{  \sigma  _{  \ottnt{B}  ,  \ottnt{A}  } }
        \ottnt{A}  \mprod  \ottnt{B}
    \end{align}
    \begin{align}
          \lambda  _{ \ottnt{A} }  =
         1_{\mode L}   \mprod  \ottnt{A}
        \xrightarrow{  \sigma  _{   1_{\mode L}   ,  \ottnt{A}  } }
        \ottnt{A}  \mprod   1_{\mode L} 
        \xrightarrow{  \rho  _{ \ottnt{A} } }
        \ottnt{A}
    \end{align}
\end{definition}

\begin{definition}\label{def:SMF}
    Let $(\cat{C}, \mprod , 1 , \alpha , \lambda , \rho , \sigma )$
    and $(\cat{C}', \mprod' ,\ottsym{1'},  \alpha  _{    }   \ottsym{'},  \lambda  _{    }   \ottsym{'},  \rho  _{    }   \ottsym{'},  \sigma  _{    }   \ottsym{'})$
    be symmetric monoidal categories. A \emph{symmetric monoidal functor}
    $\functor{F}  \ottsym{:}  \cat{C}  \Rightarrow  \cat{C}'$ is a functor along with a map
    $  m^{ \functor{F} }  _{  1  }   \ottsym{:}  \ottsym{1'}  \rightarrow  \functor{F} \,  1 $ and a natural transformation
    $  m^{ \functor{F} }  _{  \ottnt{A}  ,  \ottnt{B}  }   \ottsym{:}  \functor{F} \, \ottsym{(}  \ottnt{A}  \ottsym{)}  \mprod'  \functor{F} \, \ottsym{(}  \ottnt{B}  \ottsym{)}  \rightarrow  \functor{F} \, \ottsym{(}  \ottnt{A}  \mprod  \ottnt{B}  \ottsym{)}$ that satisfies
    the following coherence conditions:
    \centering{
    \scalebox{0.85}{\begin{tikzpicture}
        \node (top left) at (-3, 2) {$\ottsym{(}  \functor{F} \, \ottsym{(}  \ottnt{A_{{\mathrm{1}}}}  \ottsym{)}  \mprod'  \functor{F} \, \ottsym{(}  \ottnt{A_{{\mathrm{2}}}}  \ottsym{)}  \ottsym{)}  \mprod'  \functor{F} \, \ottsym{(}  \ottnt{A_{{\mathrm{3}}}}  \ottsym{)}$} ;
        \node (mid left) at (-3, 0) {$\functor{F} \, \ottsym{(}  \ottnt{A_{{\mathrm{1}}}}  \mprod  \ottnt{A_{{\mathrm{2}}}}  \ottsym{)}  \mprod'  \functor{F} \, \ottsym{(}  \ottnt{A_{{\mathrm{3}}}}  \ottsym{)}$} 
            edge [<-] node[auto] {$  m^{ \functor{F} }  _{  \ottnt{A_{{\mathrm{1}}}}  ,  \ottnt{A_{{\mathrm{2}}}}  }   \mprod'    \textrm{id}  _{    } $} (top left) ;
        \node (bot left) at (-3, -2) {$\functor{F} \, \ottsym{(}  \ottsym{(}  \ottnt{A_{{\mathrm{1}}}}  \mprod  \ottnt{A_{{\mathrm{2}}}}  \ottsym{)}  \mprod  \ottnt{A_{{\mathrm{3}}}}  \ottsym{)}$} 
            edge [<-] node[auto] {$  m^{ \functor{F} }  _{  \ottnt{A_{{\mathrm{1}}}}  \mprod  \ottnt{A_{{\mathrm{2}}}}  ,  \ottnt{A_{{\mathrm{3}}}}  } $} (mid left) ;

        \node (top right) at (3, 2) {$\functor{F} \, \ottsym{(}  \ottnt{A_{{\mathrm{1}}}}  \ottsym{)}  \mprod'  \ottsym{(}  \functor{F} \, \ottsym{(}  \ottnt{A_{{\mathrm{2}}}}  \ottsym{)}  \mprod'  \functor{F} \, \ottsym{(}  \ottnt{A_{{\mathrm{3}}}}  \ottsym{)}  \ottsym{)}$} 
            edge [<-] node[auto,swap] {$  \alpha  _{    }   \ottsym{'}$} (top left) ;
        \node (mid right) at (3, 0) {$\functor{F} \, \ottsym{(}  \ottnt{A_{{\mathrm{1}}}}  \ottsym{)}  \mprod'  \functor{F} \, \ottsym{(}  \ottnt{A_{{\mathrm{2}}}}  \mprod  \ottnt{A_{{\mathrm{3}}}}  \ottsym{)}$}
            edge [<-] node[auto,swap] {$  \textrm{id}  _{    }   \mprod'    m^{ \functor{F} }  _{  \ottnt{A_{{\mathrm{2}}}}  ,  \ottnt{A_{{\mathrm{3}}}}  } $} (top right) ;
        \node (bot right) at (3, -2) {$\functor{F} \, \ottsym{(}  \ottnt{A_{{\mathrm{1}}}}  \mprod  \ottsym{(}  \ottnt{A_{{\mathrm{2}}}}  \mprod  \ottnt{A_{{\mathrm{3}}}}  \ottsym{)}  \ottsym{)}$} 
            edge [<-] node[auto,swap] {$  m^{ \functor{F} }  _{  \ottnt{A_{{\mathrm{1}}}}  ,  \ottnt{A_{{\mathrm{2}}}}  \mprod  \ottnt{A_{{\mathrm{3}}}}  } $} (mid right)
            edge [<-] node[auto,swap] {$\functor{F} \, \ottsym{(}    \alpha  _{    }   \ottsym{)}$} (bot left) ;
    \end{tikzpicture}}\quad
    \scalebox{0.85}{\begin{tikzpicture}
        \node (top left) at (-2, 1) {$\functor{F} \, \ottsym{(}  \ottnt{A}  \ottsym{)}  \mprod'  \functor{F} \, \ottsym{(}  \ottnt{B}  \ottsym{)}$} ;
        \node (top right) at (2, 1) {$\functor{F} \, \ottsym{(}  \ottnt{B}  \ottsym{)}  \mprod'  \functor{F} \, \ottsym{(}  \ottnt{A}  \ottsym{)}$}
            edge [<-] node[auto,swap] {$  \sigma  _{    }   \ottsym{'}$} (top left) ;
        \node (bot left) at (-2,-1) {$\functor{F} \, \ottsym{(}  \ottnt{A}  \mprod  \ottnt{B}  \ottsym{)}$} 
            edge [<-] node[auto     ] {$  m^{ \functor{F} }  _{  \ottnt{A}  ,  \ottnt{B}  } $} (top left) ;
        \node (bot right) at (2,-1) {$\functor{F} \, \ottsym{(}  \ottnt{B}  \mprod  \ottnt{A}  \ottsym{)}$} 
            edge [<-] node[auto,swap] {$\functor{F} \, \ottsym{(}    \sigma  _{    }   \ottsym{)}$} (bot left) 
            edge [<-] node[auto,swap] {$  m^{ \functor{F} }  _{  \ottnt{B}  ,  \ottnt{A}  } $} (top right) ;
        \node at (0, -2) { };
    \end{tikzpicture}} 
    \scalebox{0.85}{\begin{tikzpicture}
        \node (left top left) at (-5 , 1) {$\ottsym{1'}  \mprod'  \functor{F} \, \ottsym{(}  \ottnt{A}  \ottsym{)}$} ;
        \node (left top right) at (-2, 1) {$\functor{F} \, \ottsym{(}  \ottnt{A}  \ottsym{)}$} 
            edge [<-] node[auto,swap] {$  \lambda  _{    }   \ottsym{'}$} (left top left);
        \node (left bot left) at (-5 ,-1) {$\functor{F} \, \ottsym{(}   1   \ottsym{)}  \mprod'  \functor{F} \, \ottsym{(}  \ottnt{A}  \ottsym{)}$} 
            edge [<-] node[auto     ] {$  m^{ \functor{F} }  _{  1  }   \mprod'    \textrm{id}  _{    } $} (left top left) ;
        \node (left bot right) at (-2,-1) {$\functor{F} \, \ottsym{(}   1   \mprod  \ottnt{A}  \ottsym{)}$} 
            edge [<-] node[auto,    ] {$  m^{ \functor{F} }  _{   1   ,  \ottnt{A}  } $} (left bot left) 
            edge [->] node[auto,swap] {$\functor{F} \, \ottsym{(}    \lambda  _{ \ottnt{A} }   \ottsym{)}$} (left top right);
    \end{tikzpicture} \qquad\qquad
    \begin{tikzpicture}
        \node (right top left) at ( 2, 1) {$\functor{F} \, \ottsym{(}  \ottnt{A}  \ottsym{)}  \mprod'  \ottsym{1'}$} ;
        \node (right top right) at (5, 1) {$\functor{F} \, \ottsym{(}  \ottnt{A}  \ottsym{)}$} 
            edge [<-] node[auto,swap] {$  \rho  _{    }   \ottsym{'}$} (right top left) ;
        \node (right bot left) at ( 2,-1) {$\functor{F} \, \ottsym{(}  \ottnt{A}  \ottsym{)}  \mprod'  \functor{F} \, \ottsym{(}   1   \ottsym{)}$} 
            edge [<-] node[auto     ] {$  \textrm{id}  _{    }   \mprod'    m^{ \functor{F} }  _{  1  } $} (right top left) ;
        \node (right bot right) at (5,-1) {$\functor{F} \, \ottsym{(}  \ottnt{A}  \mprod   1   \ottsym{)}$}
            edge [<-] node[auto,    ] {$  m^{ \functor{F} }  _{  \ottnt{A}  ,   1   } $} (right bot left)
            edge [->] node[auto,swap] {$\functor{F} \, \ottsym{(}    \rho  _{    }   \ottsym{)}$} (right top right) ;
    \end{tikzpicture}}
    }

    A functor $\functor{F}  \ottsym{:}  \cat{C}  \Rightarrow  \cat{C}'$ is \emph{symmetric comonoidal} if it is
    equipped with a map $  n^{ \functor{F} }  _{  1  }  : \functor{F} \,  1  \rightarrow \ottsym{1'}$ and natural transformation
    $  n^{ \functor{F} }  _{  \ottnt{A}  ,  \ottnt{B}  }   \ottsym{:}  \functor{F} \, \ottsym{(}  \ottnt{A}  \mprod  \ottnt{B}  \ottsym{)}  \rightarrow  \functor{F} \, \ottsym{(}  \ottnt{A}  \ottsym{)}  \mprod'  \functor{F} \, \ottsym{(}  \ottnt{B}  \ottsym{)}$ such that the appropriate
    (dual) diagrams commute.
\end{definition}

\begin{definition}
    Let $\functor{F}$ and $\functor{G}$ be symmetric monoidal functors $\functor{F},\functor{G}  \ottsym{:}  \cat{C}  \Rightarrow  \cat{C}'$.
    A monoidal natural transformation $\tau  \ottsym{:}  \functor{F}  \rightarrow  \functor{G}$ is a natural transformation
    satisfying
    \[
         \tau _{ \ottnt{A}  \mprod  \ottnt{B} }  \, \circ \,   m^{ \functor{F} }  _{  \ottnt{A}  ,  \ottnt{B}  }  =   m^{ \functor{G} }  _{  \ottnt{A}  ,  \ottnt{B}  }  \, \circ \, \ottsym{(}   \tau _{ \ottnt{A} }   \mprod'   \tau _{ \ottnt{B} }   \ottsym{)} 
        \qquad\text{and}\qquad
         \tau _{  1_{\mode L}  }  \, \circ \,   m^{ \functor{F} }  _{  1  }  =   m^{ \functor{G} }  _{  1  } .
    \]

    For $\functor{F}$ and $\functor{G}$ symmetric comonoidal functors, a 
    natural transformation $\tau  \ottsym{:}  \functor{F}  \rightarrow  \functor{G}$ is comonoidal if it satisfies
    the appropriate dual diagrams.
\end{definition}

\begin{definition}
    A \emph{symmetric (co-)monoidal adjunction} is an adjunction $\functor{F}  \dashv  \functor{G}$
    between symmetric \mbox{(co-)} monoidal functors $\functor{F}$ and $\functor{G}$ 
    where the unit and counit of the adjunction are symmetric (co-)monoidal
    natural transformations.
\end{definition}

\paragraph{The $ \textsc{LPC} $ model.}

Traditionally the multiplicative fragment of linear logic
is modeled by a *-autonomous category. For LPC, we use an equivalent notion
that puts the tensor $ \mprod $ and co-tensor $ \msum $ on equal footing, by modeling the 
category $ \cat{L} $ as a symmetric linearly distributive category with negation 
\cite{CockettSeely97}.

\begin{definition}
    Let $ \cat{L} $ be a category with two symmetric monoidal structures $ \mprod $ and 
    $ \msum $, and a natural transformation
    \vspace{-3mm}
    \[   \delta  _{  \ottnt{A_{{\mathrm{1}}}}  ,   \ottnt{A_{{\mathrm{2}}}}  ,  \ottnt{A_{{\mathrm{3}}}}   }   \ottsym{:}  \ottnt{A_{{\mathrm{1}}}}  \mprod  \ottsym{(}  \ottnt{A_{{\mathrm{2}}}}  \msum  \ottnt{A_{{\mathrm{3}}}}  \ottsym{)}  \rightarrow  \ottsym{(}  \ottnt{A_{{\mathrm{1}}}}  \mprod  \ottnt{A_{{\mathrm{2}}}}  \ottsym{)}  \msum  \ottnt{A_{{\mathrm{3}}}} \]
    Then $ \cat{L} $ is a \emph{symmetric linearly
    distributive category} if $ \delta $ satisfies a number of coherence
    conditions described by Cockett and Seely \cite{CockettSeely97}.

    $ \cat{L} $ is said to \emph{have negation} if there exists a map
    $ \dualize{ \ottsym{(}  \ottsym{-}  \ottsym{)} } $ on objects of $ \cat{L} $ and families of maps
    \[   \gamma^{\bot}  _{ \ottnt{A} }   \ottsym{:}   \dualize{ \ottnt{A} }   \mprod  \ottnt{A}  \rightarrow   \bot_{\mode L}  \qquad\text{and}\qquad   \gamma^{1}  _{ \ottnt{A} }   \ottsym{:}   1_{\mode L}   \rightarrow   \dualize{ \ottnt{A}  \msum  \ottnt{A} } \]
    commuting with $ \delta $ in certain ways.
\end{definition}

\begin{theorem}[Cockett and Seely] \label{thm:distr-*-autonomous}
    Symmetric linearly distributive categories with negation
    correspond to *-autonomous categories.
\end{theorem}

\begin{definition} \label{def:LPC}
    A linear/producing/consuming ($ \textsc{LPC} $) model consists of the following components:
    \begin{enumerate}
        \item A symmetric linearly distributive category $( \cat{L} , \mprod , \msum )$ with
        negation $ \dualize{ \ottsym{(}  \ottsym{-}  \ottsym{)} } $, finite products $ \aprod $ and finite coproducts $ \asum $.

        \item 
        Symmetric monoidal categories $( \cat{P} , \mprod )$ and $( \cat{C} , \msum )$
        in duality by means of contravariant functors
        $ \star{ \ottsym{(}  \ottsym{-}  \ottsym{)} }   \ottsym{:}   \cat{P}   \Rightarrow   \cat{C} $ and $ \lowerstar{ \ottsym{(}  \ottsym{-}  \ottsym{)} }   \ottsym{:}   \cat{C}   \Rightarrow   \cat{P} $,
        where $ \star{ \ottsym{(}  \ottsym{-}  \ottsym{)} } $ is monoidal and $ \lowerstar{ \ottsym{(}  \ottsym{-}  \ottsym{)} } $ is comonoidal, with natural
        isomorphisms 
        \vspace{-3mm} \[
              {\epsilon_*^*}  _{ \ottnt{C} }   \ottsym{:}   \star{ \ottsym{(}   \lowerstar{ { \ottnt{C} } }   \ottsym{)} }   \rightarrow  \ottnt{C} \qquad\text{and}\qquad   {\eta^*_*}  _{ \ottnt{P} }   \ottsym{:}  \ottnt{P}  \rightarrow   \lowerstar{ { \ottsym{(}   \star{ \ottnt{P} }   \ottsym{)} } } .
        \]

        \item Monoidal natural transformations
        $   e  _{ \ottnt{P} }  ^{\mprod}   \ottsym{:}  \ottnt{P}  \rightarrow   1_{\mode P} $ and $   d  _{ \ottnt{P} }  ^{\mprod}   \ottsym{:}  \ottnt{P}  \rightarrow  \ottnt{P}  \mprod  \ottnt{P}$
        in $ \cat{P} $ and comonoidal natural transformations
        $   e  _{ \ottnt{C} }  ^{\msum}   \ottsym{:}   \bot_{\mode C}   \rightarrow  \ottnt{C}$ and $   d  _{ \ottnt{C} }  ^{\msum}   \ottsym{:}  \ottnt{C}  \msum  \ottnt{C}  \rightarrow  \ottnt{C}$
        in $ \cat{C} $, interchanged under duality, such that:
        \begin{enumerate}[(a)]
            \item for every $\ottnt{P}$, $(\ottnt{P},   d  _{ \ottnt{P} }  ^{\mprod} ,   e  _{ \ottnt{P} }  ^{\mprod} )$
            forms a commutative comonoid in $ \cat{P} $; and 
            \item for every $\ottnt{C}$, $(\ottnt{C},   d  _{ \ottnt{C} }  ^{\msum} ,   e  _{ \ottnt{C} }  ^{\msum} )$
            forms a commutative monoid in $ \cat{C} $.
        \end{enumerate}
    
        \item Symmetric monoidal functors $ \bangfunctor{ \ottsym{-} }   \ottsym{:}   \cat{L}   \Rightarrow   \cat{P} $ and 
        $ F_\oc   \ottsym{:}   \cat{P}   \Rightarrow   \cat{L} $
        and symmetric comonoidal functors
        $ \whynotfunctor{ \ottsym{-} }   \ottsym{:}   \cat{L}   \Rightarrow   \cat{C} $ and $ F_\wn   \ottsym{:}   \cat{C}   \Rightarrow   \cat{L} $,
        which respect the dualities in that
        $ \dualize{ \ottsym{(}   F_\oc  \, \ottnt{P}  \ottsym{)} }   \simeq   F_\wn  \, \ottsym{(}   \star{ \ottnt{P} }   \ottsym{)}$
        and $ \bangfunctor{ \ottnt{A} }   \simeq   \whynotfunctor{  \dualize{ \ottnt{A} }  } $,
        and that form monoidal/comonoidal adjunctions
        $ \bangfunctor{ \ottsym{-} }   \dashv   F_\oc $ and $ F_\wn   \dashv   \whynotfunctor{ \ottsym{-} } $.
    \end{enumerate}
\end{definition}

To unpack condition (3), consider the definition of a commutative comonoid:

\begin{definition}
    Let $( \cat{P} , \mprod , 1_{\mode P} )$ be a symmetric monoidal category. 
    A commutative comonoid in $ \cat{P} $
    is an object $\ottnt{P}$ in $ \cat{P} $ along with two morphisms 
    $   e  _{    }  ^{\mprod}   \ottsym{:}  \ottnt{P}  \rightarrow   1_{\mode P} $ and 
    $   d  _{    }  ^{\mprod}   \ottsym{:}  \ottnt{P}  \rightarrow  \ottnt{P}  \mprod  \ottnt{P}$ that commute with the symmetric
    monoidal structure of $ \cat{P} $.
Dually, a commutative monoid in a symmetric monoidal category $( \cat{C} , \msum , \bot_{\mode C} )$ 
is an object $\ottnt{C}$ along with morphisms $   e  _{    }  ^{\msum}   \ottsym{:}   \bot_{\mode C}   \rightarrow  \ottnt{C}$ and 
$   d  _{    }  ^{\msum}   \ottsym{:}  \ottnt{C}  \msum  \ottnt{C}  \rightarrow  \ottnt{C}$.
\end{definition}

The commutative comonoids in $ \cat{P} $ ensure that all propositions are duplicable
in the producer category. This property is then preserved by the exponential 
decomposition $ F_\oc $, leading to the property that linear propositions of the form
$\oc  \ottnt{A}= F_\oc  \,  \bangfunctor{ \ottnt{A} } $ are similarly duplicable.

Because $ \bangfunctor{ \ottsym{-} }   \dashv   F_\oc $ forms a monoidal adjunction,
$ F_\oc $ is necessarily a strong monoidal functor~\cite{kelly1974doctrinal}, 
which implies that $ F_\oc $ is both monoidal and comonoidal. 
A similar result can be stated for $ F_\wn $.

\paragraph{$ \textsc{LPC} $ and other linear logic models.} \label{sec:othermodels}

As $ \textsc{LPC} $ is inspired by Benton's linear/non-linear paradigm, 
this section formalizes the relationship between $ \textsc{LPC} $, $ \textsc{LNL} $,
and single-category models of linear logic.
\begin{definition}[Melli{\`e}s \cite{Mellies03}]
    A linear/non-linear ($ \textsc{LNL} $) model consists of:
    \begin{inparaenum}[(1)]
        \item a symmetric monoidal closed category $ \cat{L} $;
        \item a cartesian category $ \cat{P} $; and
        \item functors $\functor{G}  \ottsym{:}   \cat{L}   \Rightarrow   \cat{P} $ and $\functor{F}  \ottsym{:}   \cat{P}   \Rightarrow   \cat{L} $
        that form a symmetric monoidal adjunction $\functor{F}  \dashv  \functor{G}$.\footnote{The
        $ \textsc{LNL} $ model given by Benton \cite{Benton94mixed} has the added stipulation
        that the cartesian category be cartesian closed, but other works have since
        disregarded this condition \cite{Mellies03}.}
    \end{inparaenum}
\end{definition}
In $ \textsc{LPC} $, because every object in $ \cat{P} $ forms a commutative
comonoid, $ \cat{P} $ is cartesian~\cite{Fox76}. Therefore:
\begin{proposition}\label{prop:LPC_is_LNL}
    Every $ \textsc{LPC} $ model is an $ \textsc{LNL} $ model.
\end{proposition}

In addition, a *-autonomous category in a linear/non-linear model induces an $ \textsc{LPC} $ triple:
\begin{proposition}
    If the category $ \cat{L} $ in an $ \textsc{LNL} $ model is *-autonomous, then
    $( \cat{L} ,  \cat{P} ,  \op{  \cat{P}  } )$ is an $ \textsc{LPC} $ model.
\end{proposition}

Next we prove that every $ \textsc{LPC} $ model contains a classical linear category
as defined by Schalk \cite{Schalk04}. This definition is just the extension
of Benton et al's linear category \cite{BentonBHP93term} to classical linear logic.

\begin{definition}[Schalk \cite{Schalk04}]\label{ModelCLL}
    A category $ \cat{L} $ is a model for classical linear logic if and only if it:
    \begin{inparaenum}[(1)]
        \item is *-autonomous;
        \item has finite products $ \aprod $ and thus finite coproducts $ \asum $; and
        \item has a linear exponential comonad $ \oc $ and thus a 
              linear exponential monad $ \wn $.
    \end{inparaenum}
\end{definition}

\begin{proposition}\label{prop:LPC_is_linear}
    The category $ \cat{L} $ from the $ \textsc{LPC} $ model is a model for classical linear logic.
\end{proposition}
\begin{proof}
    From \Theorem{distr-*-autonomous} we know that $ \cat{L} $ is *-autonomous,
    and by construction it has finite products and coproducts.
    Because the $ \textsc{LPC} $ model is also an $ \textsc{LNL} $ model, we may apply Benton's
    proof that every $ \textsc{LNL} $ model has a linear exponential comonad \cite{Benton94mixed}.
\end{proof}

\begin{proposition}
    Every model for classical linear logic forms an $ \textsc{LPC} $ category.
\end{proposition}
\begin{proof}
    Benton proved that every SMCC with a linear exponential comonad has an $ \textsc{LNL} $
    model. Because the linear category
    is *-autonomous the $ \textsc{LNL} $ model induces an $ \textsc{LPC} $ model.
\end{proof}

\paragraph{Interpretation of the Logic.}

We define an interpretation of the $ \textsc{LPC} $ logic that maps propositions to objects
in the categories, and derivations to morphisms.
For objects, the $ \interp{ - }_{  \mode{L}  } $ interpretation function maps any mode of proposition
into the linear category. The interpretation of linear propositions is straightforward,
and for persistent propositions we define
\[    \interp{ \ottnt{P} }_{\mode L}     =  F_\oc  \,  \interp{ \ottnt{P} }_{\mode P}  
    \qquad
     \interp{ \ottnt{C} }_{\mode L}     =  F_\wn  \,  \interp{ \ottnt{C} }_{\mode C} .
\]
The functions
$ \interp{ - }_{  \mode{P}  } $ and $ \interp{ - }_{  \mode{C}  } $ map propositions into the producer
and the consumer categories $ \cat{P} $ and $ \cat{C} $ respectively,
but they are defined only on the persistent propositions. 
To map producer propositions into the consumer category and vice versa,
we define
\[
     \interp{ \ottnt{C} }_{\mode P}  =  \lowerstar{ { \ottsym{(}   \interp{ \ottnt{C} }_{\mode C}   \ottsym{)} } }  
    \qquad
     \interp{ \ottnt{P} }_{\mode C}  =  \star{ \ottsym{(}   \interp{ \ottnt{P} }_{\mode P}   \ottsym{)} } .
\]

Linear contexts are interpreted as a single proposition in the linear category.
The comma is represented by the tensor connector $ \mprod $ if the context
is meant to appear on the left-hand-side of a sequent, and by the cotensor
$ \msum $ if the context is meant to appear on the right. These interpretations
of linear contexts are represented as
$ \interp{ \Gamma }_{\mode L}^{\mprod} $ and $ \interp{ \Delta }_{\mode L}^{\msum} $ respectively.
In the producer category there is no cotensor
and vice versa for the consumer category, so
$ \interp{  \Gamma ^{\mode P}  }_{\mode P} $ interprets the comma as the tensor 
in the producer category, and
$ \interp{  \Gamma ^{\mode C}  }_{\mode C} $ interprets the comma as the cotensor
in the consumer category.

In this way a linear derivation $\mathcal{D}$ of the form $\Gamma  \shows  \Delta$ will be interpreted
as a morphism $ \interp{ \mathcal{D} }_{\mode L}   \ottsym{:}   \interp{ \Gamma }_{\mode L}^{\mprod}   \rightarrow   \interp{ \Delta }_{\mode L}^{\msum} $. 
However, it is not clear in which category we should interpret a persistent
sequent of the form $\Gamma  \Vdash  \Delta$, since $\Gamma$ and $\Delta$ may contain 
both producer and consumer propositions. 
Recall \Proposition{displacement}, which states that every such derivation $\mathcal{D}$
contains exactly one displaced proposition. This means that $\mathcal{D}$ is either of the form
$ \Gamma ^{\mode P}   \Vdash    \Delta ^{\mode C}  , \ottnt{P} $ or $  \Gamma ^{\mode P}  , \ottnt{C}   \Vdash   \Delta ^{\mode C} $. 
In the category $ \cat{P} $, this derivation will be interpreted as a morphism
\vspace{-3mm}\begin{align*}
 \interp{ \mathcal{D} }_{\mode P}   \ottsym{:}   \interp{  \Gamma ^{\mode P}  }_{\mode P}   \mprod   \interp{  \Delta ^{\mode C}  }_{\mode P}   \rightarrow   \interp{ \ottnt{P} }_{\mode P}  \quad\text{or} \quad
 \interp{ \mathcal{D} }_{\mode P}   \ottsym{:}   \interp{  \Gamma ^{\mode P}  }_{\mode P}   \mprod   \interp{  \Delta ^{\mode C}  }_{\mode P}   \rightarrow   \interp{ \ottnt{C} }_{\mode P} ,
\end{align*}
respectively. In the same way every derivation can be interpreted
as a morphism in $ \cat{C} $.

The interpretation is defined by mutual induction on the derivations.
\begin{enumerate}

\item The interpretation of the linear inference rules given in 
Figure~\ref{fig:linrules-linear} as well as the persistent rules
in Figure~\ref{fig:linrules-persistent} are straightforward from
the categorical structures.

\item 
\begin{minipage}[t]{0.70\textwidth}
The interpretation of weakening and contraction rules is defined
using the monoid in $ \cat{C} $ and comonoid in $ \cat{P} $. 
For weakening in the linear sequent, suppose $\mathcal{D}$ is the
derivation to the right.
The interpretation of $\mathcal{D}$ inserts the comonoidal component $   e  _{    }  ^{\mprod} $ in 
$ \cat{P} $ into the linear category:
\end{minipage}
\begin{minipage}[t]{0.25\textwidth}{
\vspace{-5mm}
\[
    \mathcal{D} = \inferrule* [right=$   \textsc{W}  ^{\vdash}  \text{-L} $]
    { \inferrule* { \mathcal{D}' } { \Gamma  \shows  \Delta } }
    { \Gamma , \ottnt{P}   \shows  \Delta}
\] }
\end{minipage}
\vspace{-5mm}
\begin{align*}
     \interp{ \mathcal{D} }_{\mode L}  :
     \interp{ \Gamma }_{\mode L}^{\mprod}   \mprod   F_\oc  \,  \interp{ \ottnt{P} }_{\mode P}  
    &\xrightarrow{ \interp{ \mathcal{D}' }_{\mode L}   \mprod   F_\oc  \,  {    e  _{    }  ^{\mprod}  } }
     \interp{ \Delta }_{\mode L}^{\msum}   \mprod   F_\oc  \,  1_{\mode P}  
    \xrightarrow{  \textrm{id}  _{    }   \mprod   {   n^{  F_\oc  }  _{    }  } }
     \interp{ \Delta }_{\mode L}^{\msum}   \mprod   1_{\mode L}  
    \xrightarrow{   \rho  _{    }  ^{\mprod} }
     \interp{ \Delta }_{\mode L}^{\msum} 
\end{align*}

\item
\begin{minipage}[t]{0.70\textwidth}
If the last rule in the derivation is an $   F_\oc   \text{-L} $ or $   F_\wn   \text{-R} $
rule, its interpretation is just the interpretation of its subderivation.
On the other hand, if the last rule is the right $ F_\oc $ rule, the
inductive hypothesis states that there exists a morphism 
$ \interp{ \mathcal{D}' }_{\mode P}   \ottsym{:}   \interp{  \Gamma ^{\mode P}  }_{\mode P}   \mprod   \interp{  \Delta ^{\mode C}  }_{\mode P}   \rightarrow   \interp{ \ottnt{P} }_{\mode P} $.
It is necessary to undo this duality transformation for interpretation in the linear category.
\end{minipage}
\begin{minipage}[t]{0.25\textwidth}
\vspace{-5mm}
\[
    \mathcal{D} = \inferrule* [right=$   F_\oc   \text{-R} $]
    { \inferrule* { \mathcal{D}' } {  \Gamma ^{\mode P}    \Vdash     \Delta ^{\mode C}  , \ottnt{P}  } }
    { \Gamma ^{\mode P}   \shows    \Delta ^{\mode C}  ,  F_\oc  \, \ottnt{P} }
\]
\end{minipage}

Notice that for any persistent context $\Gamma$, there is an isomorphism 
$ \pi  :  \interp{ \Gamma }_{\mode L}^{\mprod}  \cong  F_\oc  \,  \interp{ \Gamma }_{\mode P} $ given by the monoidal components of $ F_\oc $.
Furthermore, there is an isomorphism $ \tau $ between $ \dualize{ \ottsym{(}   \interp{ \Delta }_{\mode L}^{\msum}   \ottsym{)} } $ and $ F_\oc  \,  \interp{ \Delta }_{\mode P} $ 
given by the isomorphism $ \dualize{ \ottsym{(}   F_\wn  \, \ottnt{C}  \ottsym{)} }  \cong  F_\oc  \,  \lowerstar{ { \ottnt{C} } } $. Using $ \pi $ and $ \tau $ we define
the interpretation of $\mathcal{D}$:
{\small  \[ \hspace{-5mm} \begin{array}{rclcl}
     \interp{ \mathcal{D} }_{\mode L}  : 
     \interp{  \Gamma ^{\mode P}  }_{\mode L}^{\mprod}  
    &\xrightarrow{   \rho  _{    }  ^{\mprod} ;\ottsym{(}    \textrm{id}  _{    }   \mprod    \gamma^{1}  _{    }   \ottsym{)}}
    & \interp{  \Gamma ^{\mode P}  }_{\mode L}^{\mprod}   \mprod  \ottsym{(}   \dualize{ \ottsym{(}   \interp{  \Delta ^{\mode C}  }_{\mode L}^{\msum}   \ottsym{)} }   \msum   \interp{  \Delta ^{\mode C}  }_{\mode L}^{\msum}   \ottsym{)} 
    &\xrightarrow{  \pi  _{    }   \mprod  \ottsym{(}    \tau  _{    }   \msum    \textrm{id}  _{    }   \ottsym{)}} 
    & F_\oc  \,  \interp{  \Gamma ^{\mode P}  }_{\mode P}   \mprod  \ottsym{(}   F_\oc  \,  \interp{  \Delta ^{\mode C}  }_{\mode P}   \msum   \interp{  \Delta ^{\mode C}  }_{\mode L}^{\msum}   \ottsym{)} \\
    &\xrightarrow{ \delta }
    &\ottsym{(}   F_\oc  \,  \interp{  \Gamma ^{\mode P}  }_{\mode P}   \mprod   F_\oc  \,  \interp{  \Delta ^{\mode C}  }_{\mode P}   \ottsym{)}  \msum   \interp{  \Delta ^{\mode C}  }_{\mode L}^{\msum}  
    &\xrightarrow{  m^{  F_\oc  }  _{    }   \msum    \textrm{id}  _{    } }
    & F_\oc  \, \ottsym{(}   \interp{  \Gamma ^{\mode P}  }_{\mode P}   \mprod   \interp{  \Delta ^{\mode C}  }_{\mode P}   \ottsym{)}  \msum   \interp{  \Delta ^{\mode C}  }_{\mode L}^{\msum}  \\
    &\xrightarrow{ F_\oc  \,  \interp{ \mathcal{D}' }_{\mode P}   \msum    \textrm{id}  _{    } }
    & F_\oc  \,  \interp{ \ottnt{P} }_{\mode P}   \msum   \interp{  \Delta ^{\mode C}  }_{\mode L}^{\msum}  
    &\xrightarrow{   \sigma  _{    }  ^{\msum} }
    & \interp{  \Delta ^{\mode C}  }_{\mode L}^{\msum}   \msum   \interp{  F_\oc  \, \ottnt{P} }_{\mode L} 
\end{array} \] }

\item
\begin{minipage}[t]{0.70\textwidth}
Suppose the last rule in $\mathcal{D}$ is the left $ \bangfunctor{ \ottsym{-} } $ rule.
The interpretation of $\mathcal{D}$ should be a morphism from 
$ \interp{ \Gamma }_{\mode L}^{\mprod}   \mprod   F_\oc  \, \ottsym{(}   \bangfunctor{  \interp{ \ottnt{A} }_{\mode L}  }   \ottsym{)}$
to $ \interp{ \Delta }_{\mode L}^{\msum} $; we use the unit of the adjunction,
$  \epsilon  _{    }   \ottsym{:}   F_\oc  \,  \bangfunctor{ \ottnt{A} }   \rightarrow  \ottnt{A}$ to cancel out the exponentials.
\begin{align*}
     \interp{ \mathcal{D} }_{\mode L}  :
     \interp{ \Gamma }_{\mode L}^{\mprod}   \mprod   F_\oc  \, \ottsym{(}   \bangfunctor{  \interp{ \ottnt{A} }_{\mode L}  }   \ottsym{)} 
    \xrightarrow{  \textrm{id}  _{    }   \mprod    \epsilon  _{    } }
     \interp{ \Gamma }_{\mode L}^{\mprod}   \mprod   \interp{ \ottnt{A} }_{\mode L}  
    \xrightarrow{ \interp{ \mathcal{D}' }_{\mode L} }
     \interp{ \Delta }_{\mode L}^{\msum} 
\end{align*}

Similarly, the $   \bangfunctor{ \ottsym{-} }   \text{-R} $ rule uses the counit of the adjunction, along
with the isomorphisms $ \pi $ and $ \tau $ defined previously.
If the last rule in $\mathcal{D}$ is the $   \bangfunctor{ \ottsym{-} }   \text{-R} $ rule, its interpretation
is defined as follows:
\end{minipage}
\begin{minipage}[t]{0.25\textwidth}
\vspace{-8mm}
\begin{align*}
    \mathcal{D} = \inferrule* [right=$   \bangfunctor{ \ottsym{-} }   \text{-L} $]
    { \inferrule* { \mathcal{D}' } {  \Gamma , \ottnt{A}   \shows  \Delta } }
    { \Gamma ,  \bangfunctor{ \ottnt{A} }    \shows  \Delta}
\end{align*}
\begin{align*}
    \mathcal{D} = \inferrule* [right=$   \bangfunctor{ \ottsym{-} }   \text{-R} $]
    { \inferrule* { \mathcal{D}' } {  \Gamma ^{\mode P}   \shows    \Delta ^{\mode C}  , \ottnt{A}  } }
    { \Gamma ^{\mode P}   \Vdash    \Delta ^{\mode C}  ,  \bangfunctor{ \ottnt{A} }  }
\end{align*}
\end{minipage}

{\small \[ \hspace{-10mm} \begin{array}{rclcl}
     \interp{ \mathcal{D} }_{\mode P}  :  \interp{  \Gamma ^{\mode P}  }_{\mode P}   \mprod   \interp{  \Delta ^{\mode C}  }_{\mode P} 
    &\xrightarrow{  \eta  _{    }   \mprod    \eta  _{    } }
    & \bangfunctor{  F_\oc  \,  \interp{  \Gamma ^{\mode P}  }_{\mode P}  }   \mprod   \bangfunctor{  F_\oc  \,  \interp{  \Delta ^{\mode C}  }_{\mode P}  } 
    &\xrightarrow{ m^{  \bangfunctor{ \ottsym{-} }  } }
    & \bangfunctor{  F_\oc  \,  \interp{  \Gamma ^{\mode P}  }_{\mode P}   \mprod   F_\oc  \,  \interp{  \Delta ^{\mode C}  }_{\mode P}  }  \\
    &\xrightarrow{ \bangfunctor{  \inv{   \pi  _{    }  }   \mprod   {  \inv{   \tau  _{    }  }  }  } } 
    & \bangfunctor{  \dualize{  \interp{  \Gamma ^{\mode P}  }_{\mode L}^{\mprod}   \mprod  \ottsym{(}   \interp{  \Delta ^{\mode C}  }_{\mode L}^{\msum}   \ottsym{)} }  }  
    &\xrightarrow{ \bangfunctor{  \interp{ \mathcal{D}' }_{\mode L}   \mprod    \textrm{id}  _{    }  } }
    & \bangfunctor{  \dualize{ \ottsym{(}   \interp{  \Delta ^{\mode C}  }_{\mode L}^{\msum}   \msum   \interp{ \ottnt{A} }_{\mode L}   \ottsym{)}  \mprod  \ottsym{(}   \interp{  \Delta ^{\mode C}  }_{\mode L}^{\msum}   \ottsym{)} }  }  \\
    &\xrightarrow{ \bangfunctor{   \delta  _{    }  } }
    & \bangfunctor{ \ottsym{(}   \dualize{  \interp{  \Delta ^{\mode C}  }_{\mode L}^{\msum}   \mprod  \ottsym{(}   \interp{  \Delta ^{\mode C}  }_{\mode L}^{\msum}   \ottsym{)} }   \ottsym{)}  \msum   \interp{ \ottnt{A} }_{\mode L}  }  
    &\xrightarrow{  \gamma^{\bot}  _{    }   \msum    \textrm{id}  _{    } ; \bangfunctor{    \lambda  _{    }  ^{\msum}  } }
    & \bangfunctor{  \interp{ \ottnt{A} }_{\mode L}  }  =  \interp{  \bangfunctor{ \ottnt{A} }  }_{\mode P} 
\end{array} \] }

\end{enumerate}

\section{Examples}
\label{sec:examples}
This section provides some concrete instances of the $ \textsc{LPC} $ model.
The following chart summarizes the three examples and their LPC categories.
\vspace{-4mm}
\[ \begin{array}{l c c c}
&     \cat{L}  &  \cat{P}  &  \cat{C}  \\
    \hline 
\rule{0pt}{12pt}\mbox{Vectors} &     \FinVect  &  \FinSet  &  \op{  \FinSet  }  \\
\mbox{Relations} &   \Rel  &  \Set  &  \op{  \Set  }  \\
\mbox{Bool. Alg.} &     \FinBoolAlg  &  \FinPoset  &  \FinLat  
\end{array} \]

\paragraph{Vector Spaces.}
Linear logic shares many features with linear algebra, based on the natural 
interpretations of the tensor product and duality of vector spaces.
To construct an $ \textsc{LPC} $ model, let $ \cat{L} $ be the category of finite-dimensional
vector spaces over a finite field $\F$, $ \cat{P} $ be the category of finite sets and functions,
and $ \cat{C} $ be the opposite category of $ \cat{P} $.

The $ \mprod $ operator of linear logic is easily interpreted as the tensor product in $ \cat{L} $.
The $ \msum $ operator has no natural interpretation in terms of vector spaces,
but we may define $\ottnt{U}  \msum  \ottnt{V} := \ottnt{U}  \mprod  \ottnt{V}$.
The units $ 1_{\mode L} $ and $ \bot_{\mode L} $ may be any one-dimensional vector space; for 
concreteness let them be generated by the basis $\{  \mathbb{1}  \}$.

    The free vector space $\ottkw{Free} \, \ottsym{(}  X  \ottsym{)}$ of a finite set $X$ over $\F$ is the vector space
    with vectors the formal sums $\alpha_1 x_1 + \cdots + \alpha_n x_n$,
    addition defined pointwise, and scalar multiplication defined by distribution over the $x_i$'s.
A basis for $\ottkw{Free} \, \ottsym{(}  X  \ottsym{)}$ is the set $ \{   \delta_{ \ottnt{x} }   \mid   \ottnt{x}  \in  X   \} $ where $ \delta_{ \ottnt{x} } $ is the free sum $\ottnt{x}$.

    The dual of a vector space $\ottnt{V}$ (with basis $ B $)
    over $\F$ is the set $ \dualize{ \ottnt{V} } $ of linear maps from $\ottnt{V}$ to $\F$.
For any vector $ \ottnt{v}  \in  \ottnt{V} $, we can define $  \overline{ \ottnt{v} }   \in   \dualize{ \ottnt{V} }  $ to be the linear map 
acting on basis elements $ \ottnt{x}  \in  \ottnt{B} $ by
\[  \overline{ \ottnt{v} }   \ottsym{[}  \ottnt{x}  \ottsym{]} = \begin{cases}
         1  & \ottnt{x}  \ottsym{=}  \ottnt{v} \\
         0  &  \ottnt{x}  \neq  \ottnt{v} 
    \end{cases}
\]
Addition and scalar multiplication are defined pointwise.
Then $ \{   \overline{ \ottnt{x} }   \mid   \ottnt{x}  \in   B    \} $ is a basis for $ \dualize{ \ottnt{V} } $.

The additives $ \aprod $ and $ \asum $ are embodied by the notions of the direct product and
direct sum, which in the case of finite-dimensional vector spaces, coincide.

\begin{lemma}
    The category $ \FinVect $ is a symmetric linearly distributive category with negation,
    products, and coproducts.
\end{lemma}
\begin{proof}
    Since $ \mprod $ and $ \msum $ overlap, the distributivity transformation $ \delta $ is simply
    associativity. The coherence diagrams for linear distribution then depend on 
    the commutativity of tensor, associativity, and swap morphisms.
    To show the category has negation, we define $  \gamma^{\bot}  _{ \ottnt{A} }   \ottsym{:}   \dualize{ \ottnt{A} }   \mprod  \ottnt{A}  \rightarrow   \bot $ and 
    $  \gamma^{1}  _{ \ottnt{A} }   \ottsym{:}   1   \rightarrow   \dualize{ \ottnt{A}  \msum  \ottnt{A} } $ as follows,
    where $ B $ is a basis for $\ottnt{A}$:
    \[
          \gamma^{\bot}  _{ \ottnt{A} }  \, \ottsym{(}    \delta_{ \ottnt{u} }    \mprod  \ottnt{v}  \ottsym{)} =   \delta_{ \ottnt{u} }   \ottsym{[}  \ottnt{v}  \ottsym{]}   \cdot   \mathbb{1} 
        \qquad\qquad
          \gamma^{1}  _{ \ottnt{A} }  \, \ottsym{(}   \mathbb{1}   \ottsym{)}  \ottsym{=}   \sum_{ \ottnt{v}  \in   B  }   \ottnt{v}  \mprod   \overline{ \ottnt{v} }   
    \]
    It then suffices to check that
    $  \lambda  _{    }  \, \circ \, \ottsym{(}    \gamma^{\bot}  _{    }   \mprod    \textrm{id}  _{    }   \ottsym{)} \, \circ \,   \alpha  _{    }  \, \circ \, \ottsym{(}    \textrm{id}  _{    }   \mprod    \lambda  _{    }   \ottsym{)} =   \rho  _{    } $.
\end{proof}

We will present only the adjunction between $ \FinVect $ and the producing
category $ \FinSet $; the other can be inferred from the opposite category.
Define $ \bangfunctor{ \ottsym{-} }   \ottsym{:}   \FinVect   \Rightarrow   \FinSet $ to be the forgetful functor, which takes a vector space
to its underlying set of vectors.
It is a monoidal functor with components $  m^{  \bangfunctor{ \ottsym{-} }  }  _{  1  }   \ottsym{:}   1_{\mode P}   \rightarrow   \bangfunctor{  1  } $
and $  m^{  \bangfunctor{ \ottsym{-} }  }  _{  \ottnt{A}  ,  \ottnt{B}  }   \ottsym{:}   \bangfunctor{ \ottnt{A} }  \, \times \,  \bangfunctor{ \ottnt{B} }   \rightarrow   \bangfunctor{ \ottnt{A}  \mprod  \ottnt{B} } $ defined by
    $  m^{  \bangfunctor{ \ottsym{-} }  }  _{  1_{\mode L}  }  \, \ottsym{(}    \emptyset    \ottsym{)} =  \mathbb{1} $
    and 
    $  m^{  \bangfunctor{ \ottsym{-} }  }  _{  \ottnt{A}  ,  \ottnt{B}  }  \, \ottsym{(}  \ottnt{u}  \ottsym{,}  \ottnt{v}  \ottsym{)}  \ottsym{=}  \ottnt{u}  \mprod  \ottnt{v}$.

On objects, the functor $ F_\oc   \ottsym{:}   \FinSet   \Rightarrow   \FinVect $ takes a set $X$ to the free vector
space generated by $X$. For a morphism $\ottnt{f}  \ottsym{:}  X_{{\mathrm{1}}}  \rightarrow  X_{{\mathrm{2}}}$ in $ \FinSet $, we define
$ F_\oc  \, \ottnt{f} : \ottkw{Free} \, \ottsym{(}  X_{{\mathrm{1}}}  \ottsym{)} \rightarrow \ottkw{Free} \, \ottsym{(}  X_{{\mathrm{2}}}  \ottsym{)}$ to be
$ F_\oc  \, \ottnt{f} \, \ottsym{(}   \delta_{ \ottnt{x} }   \ottsym{)} =  \delta_{  \ottnt{f} \, \ottsym{(}  \ottnt{x}  \ottsym{)}  } .$
Then $ F_\oc $ is monoidal with components $  m^{  F_\oc  }  _{  1  }   \ottsym{:}   1   \rightarrow   F_\oc  \,  1_{\mode P} $
and $  m^{  F_\oc  }  _{  X_{{\mathrm{1}}}  ,  X_{{\mathrm{2}}}  }   \ottsym{:}   F_\oc  \, X_{{\mathrm{1}}}  \mprod   F_\oc  \, X_{{\mathrm{2}}}  \rightarrow   F_\oc  \, \ottsym{(}  X_{{\mathrm{1}}} \, \times \, X_{{\mathrm{2}}}  \ottsym{)}$ defined~as
\[
      m^{  F_\oc  }  _{  1_{\mode L}  }  \, \ottsym{(}   \mathbb{1}   \ottsym{)} =  \delta_{   \emptyset   }  
    \qquad\qquad
      m^{  F_\oc  }  _{  X_{{\mathrm{1}}}  ,  X_{{\mathrm{2}}}  }  \, \ottsym{(}    \delta_{ \ottnt{x_{{\mathrm{1}}}} }    \mprod    \delta_{ \ottnt{x_{{\mathrm{2}}}} }    \ottsym{)}  \ottsym{=}   \delta_{ \ottsym{(}  \ottnt{x_{{\mathrm{1}}}}  \ottsym{,}  \ottnt{x_{{\mathrm{2}}}}  \ottsym{)} } 
\]

\begin{lemma}
    The functors $ \bangfunctor{ \ottsym{-} } $ and $ F_\oc $ form a symmetric monoidal adjunction $ \bangfunctor{ \ottsym{-} }   \dashv   F_\oc $.
\end{lemma}
\begin{proof}
    We define the unit $  \epsilon  _{ \ottnt{A} }   \ottsym{:}   F_\oc  \,  \bangfunctor{ \ottnt{A} }   \rightarrow  \ottnt{A}$ and counit $  \eta  _{ \ottnt{P} }   \ottsym{:}  \ottnt{P}  \rightarrow   \bangfunctor{  F_\oc  \, \ottnt{P} } $ 
    of the adjunction as follows: 
    \[
          \epsilon  _{ \ottnt{A} }  \, \ottsym{(}   \delta_{ \ottnt{v} }   \ottsym{)}  \ottsym{=}  \ottnt{v} \qquad\qquad   \eta  _{ \ottnt{P} }  \, \ottsym{(}  \ottnt{x}  \ottsym{)}  \ottsym{=}   \delta_{ \ottnt{x} } 
    \]
    It is easy to check that $ \epsilon $ and $ \eta $ form an adjunction,
    and are both monoidal natural transformations.
\end{proof}

\begin{corollary} 
    $ \FinVect $, $ \FinSet $, and $ \op{  \FinSet  } $ together form an $ \textsc{LPC} $ model.
\end{corollary}

Linear algebra has been considered as a model for linear logic multiple times in 
the literature.
    Ehrhard \cite{Ehrhard05} presents finiteness
     spaces, where the objects are spaces of vectors with finite support. In his model,
    the $ \oc $ operator sends a space $\ottnt{A}$ to the space supported
    by finite multisets over $\ottnt{A}$;
    it takes some effort to show that this comonad respects the finiteness conditions.
    Pratt \cite{Pratt94} proves that finite dimensional vector spaces over a field 
    of characteristic~2 is a Chu space and thus a model of linear logic. 
    Valiron and Zdancewic \cite{ValironZdancewic14} show that the $ \textsc{LPC} $ 
    model of $ \FinVect $ is
    a sound and complete semantic model for an algebraic $\lambda$-calculus.    

\paragraph{Relations.}
Let $ \Rel $ be the category of sets and relations, and let $ \Set $ be the category
of sets and functions. (Notice that the sets in either category here may be infinite, unlike
in the $ \FinVect $ case.) It is easy to see that $ \Rel $ is linearly distributive where
the tensor and the cotensor are both cartesian product, and distributivity is
just associativity. The unit is a singleton set, and
negation on $ \Rel $ is the identity operation.

$ \Set $ is cartesian and its opposite category $ \op{  \Set  } $, cocartesian.
The $ F_\oc $ and $ F_\wn $ functors are the forgetful functors which interpret a
function as a relation. The $ \bangfunctor{ \ottsym{-} } $ functor takes a set to its powerset.
Suppose $ R $ is a relation between $\ottnt{A}$ and $\ottnt{B}$. The function
$ \bangfunctor{  R  }   \ottsym{:}   \bangfunctor{ \ottnt{A} }   \rightarrow   \bangfunctor{ \ottnt{B} } $ is defined as
\begin{align*}
     \bangfunctor{  R  } \ottsym{(}  X  \ottsym{)} &=  \{   \ottnt{y}  \in  \ottnt{B}   \mid   \exists   \ottnt{x}  \in  X  ,  \ottsym{(}  \ottnt{x}  \ottsym{,}  \ottnt{y}  \ottsym{)}  \in   R     \} .
\end{align*}
Then $ \bangfunctor{ \ottsym{-} } $ has monoidal components 
$  m^{  \bangfunctor{ \ottsym{-} }  }  _{  1  }   \ottsym{:}   1_{\mode P}   \rightarrow   \bangfunctor{  1_{\mode L}  } $ and $  m^{  \bangfunctor{ \ottsym{-} }  }  _{  \ottnt{A}  ,  \ottnt{B}  }   \ottsym{:}   \bangfunctor{ \ottnt{A} }  \, \times \,  \bangfunctor{ \ottnt{B} }   \rightarrow   \bangfunctor{ \ottnt{A} \, \times \, \ottnt{B} } $ defined by
\[
      m^{  \bangfunctor{ \ottsym{-} }  }  _{  1_{\mode L}  }  \, \ottsym{(}    \emptyset    \ottsym{)}
        =    \emptyset    
    \qquad
      m^{  \bangfunctor{ \ottsym{-} }  }  _{  \ottnt{A}  ,  \ottnt{B}  }  \, \ottsym{(}   X_{{\mathrm{1}}}   \ottsym{,}   X_{{\mathrm{2}}}   \ottsym{)}
        = X_{{\mathrm{1}}} \, \times \, X_{{\mathrm{2}}}
\]
The dual notion $ \whynotfunctor{ \ottsym{-} } $ is just the inverse.

Melli{\'e}s \cite{Mellies03} discusses a non-model of linear logic based on $ \Rel $,
where the exponential takes a set $X$ to the finite subsets of $X$.
That ``model'' fails because the comonad unit $  \epsilon  _{ \ottnt{A} }   \ottsym{:}  \oc  \ottnt{A}  \rightarrow  \ottnt{A}$ is not natural. 
In the $ \textsc{LPC} $ formulation, $ \epsilon $ is derived from the adjunction,
ensuring naturality.

\paragraph{Boolean Algebras.}

Next we consider an example of the $ \textsc{LPC} $ categories where $ \cat{P} $
and $ \cat{C} $ are related by a non-trivial duality. The relationship
is based on Birkhoff's representation theorem \cite{Birkhoff37}, which 
can be interpreted as a duality between the categories of finite partial orders
and order-preserving maps ($ \cat{P} $) on the one hand, and 
finite distributive lattices with bounded lattice homomorphisms ($ \cat{C} $) 
on the other hand. 

The linear category $ \cat{L} $ is the category of
finite boolean algebras with bounded lattice homomorphisms.
For the monoidal structure, the units are both the singleton lattice $ \{    \emptyset    \} $,
and the tensors $\ottnt{A}  \mprod  \ottnt{B}$ and $\ottnt{A}  \msum  \ottnt{B}$ are the boolean algebra with base set
$\ottnt{A} \, \times \, \ottnt{B}$ and lattice structure as follows:
\[ \begin{aligned}
     \bot  &= \ottsym{(}   \bot   \ottsym{,}   \bot   \ottsym{)} \\
    \ottsym{(}  \ottnt{x_{{\mathrm{1}}}}  \ottsym{,}  \ottnt{y_{{\mathrm{1}}}}  \ottsym{)}  \vee  \ottsym{(}  \ottnt{x_{{\mathrm{2}}}}  \ottsym{,}  \ottnt{y_{{\mathrm{2}}}}  \ottsym{)} &= \ottsym{(}  \ottnt{x_{{\mathrm{1}}}}  \vee  \ottnt{x_{{\mathrm{2}}}}  \ottsym{,}  \ottnt{y_{{\mathrm{1}}}}  \vee  \ottnt{y_{{\mathrm{2}}}}  \ottsym{)} 
\end{aligned} \qquad \begin{aligned}
     \top  &= \ottsym{(}   \top   \ottsym{,}   \top   \ottsym{)} \\
    \ottsym{(}  \ottnt{x_{{\mathrm{1}}}}  \ottsym{,}  \ottnt{y_{{\mathrm{1}}}}  \ottsym{)}  \wedge  \ottsym{(}  \ottnt{x_{{\mathrm{2}}}}  \ottsym{,}  \ottnt{y_{{\mathrm{2}}}}  \ottsym{)} &= \ottsym{(}  \ottnt{x_{{\mathrm{1}}}}  \wedge  \ottnt{x_{{\mathrm{2}}}}  \ottsym{,}  \ottnt{y_{{\mathrm{1}}}}  \wedge  \ottnt{y_{{\mathrm{2}}}}  \ottsym{)}
\end{aligned} \qquad
\begin{aligned}
     \neg  \ottsym{(}  \ottnt{x}  \ottsym{,}  \ottnt{y}  \ottsym{)}  = \ottsym{(}   \neg  \ottnt{x}   \ottsym{,}   \neg  \ottnt{y}   \ottsym{)}
\end{aligned}
\]

Given a partially ordered set $(\ottnt{P}, \le )$, a subset $X \subseteq \ottnt{P}$
is called \emph{lower} if it is downwards closed with respect to $ \le $.
The set of all lower sets of $\ottnt{P}$ forms a lattice with $ \top =\ottnt{P}$,
$ \bot_{\mode L} = \emptyset $, meet as union and join and intersection. Let $ \star{ \ottnt{P} } $
refer to this lattice.

Meanwhile, given a lattice $\ottnt{C}$, an element $\ottnt{x}$ is join-irreducible 
if $\ottnt{x}$ is neither $ \bot_{\mode L} $ nor the join of any two elements less than $\ottnt{x}$.
That is, $\ottnt{x} \neq \ottnt{y}  \vee  \ottnt{z}$ for $\ottnt{y},\ottnt{z} \neq \ottnt{x}$. Let 
$ \lowerstar{ { \ottnt{C} } } $ be the partially ordered set with base set the join-irreducible elements
of $\ottnt{C}$, with the ordering 
$\ottnt{x}  \le  \ottnt{y}$ iff $\ottnt{x}  \ottsym{=}  \ottnt{y}  \wedge  \ottnt{x}$.

The operators $ \star{ \ottsym{(}  \ottsym{-}  \ottsym{)} } $ and $ \lowerstar{ \ottsym{(}  \ottsym{-}  \ottsym{)} } $ extend to functors
that form a duality between $ \cat{P} $ and $ \cat{C} $ \cite{Stanley11}.

The monoidal structure on $ \cat{P} $ is given by the cartesian product with the ordering
$\ottsym{(}  \ottnt{x_{{\mathrm{1}}}}  \ottsym{,}  \ottnt{y_{{\mathrm{1}}}}  \ottsym{)}  \le  \ottsym{(}  \ottnt{x_{{\mathrm{2}}}}  \ottsym{,}  \ottnt{y_{{\mathrm{2}}}}  \ottsym{)}$ iff $\ottnt{x_{{\mathrm{1}}}}  \le  \ottnt{x_{{\mathrm{2}}}}$ and $\ottnt{y_{{\mathrm{1}}}}  \le  \ottnt{y_{{\mathrm{2}}}}$.
The unit is the singleton order $ \{    \emptyset    \} $.
It is easy to check that every poset has a communitive comonoid.

Finite distributive lattices have a monoidal structure 
with the unit the singleton lattice $ \{    \emptyset    \} $
and the tensor $\ottnt{C_{{\mathrm{1}}}}  \msum  \ottnt{C_{{\mathrm{2}}}}$ the lattice where the base set is $\ottnt{C_{{\mathrm{1}}}} \, \times \, \ottnt{C_{{\mathrm{2}}}}$.
For every lattice $\ottnt{C}$ in $ \cat{C} $ there exists a commutative monoid with
components $   e  _{ \ottnt{C} }  ^{\msum}   \ottsym{:}   \bot_{\mode C}   \rightarrow  \ottnt{C}$ and $   d  _{ \ottnt{C} }  ^{\msum}   \ottsym{:}  \ottnt{C}  \msum  \ottnt{C}  \rightarrow  \ottnt{C}$
as follows:
\[ 
       e  _{ \ottnt{C} }  ^{\msum}  \, \ottsym{(}    \emptyset    \ottsym{)} =  \bot 
    \qquad\qquad
       d  _{ \ottnt{C} }  ^{\msum}  \, \ottsym{(}  \ottnt{x}  \ottsym{,}  \ottnt{y}  \ottsym{)} = \ottnt{x}  \wedge  \ottnt{y}
\]

\noindent
The components of the monoid in $ \cat{C} $ and the comonoid in $ \cat{P} $
are interchanged under the Birkhoff duality.

Next we define the symmetric monoidal functors. Define
$ \bangfunctor{ \ottsym{-} }   \ottsym{:}   \cat{L}   \Rightarrow   \cat{P} $ and
$ \whynotfunctor{ \ottsym{-} }   \ottsym{:}   \cat{L}   \Rightarrow   \cat{C} $
to be forgetful functors.
For $ \bangfunctor{ \ottsym{-} } $ in particular, the order induced by the boolean algebra is 
$\ottnt{x}  \le  \ottnt{y}$ iff $\ottnt{x}  \ottsym{=}  \ottnt{y}  \wedge  \ottnt{x}$.

Define $ F_\oc $ and $ F_\wn $ to be the powerset algebra, which takes a structure with base
set $\ottnt{X}$ to the boolean algebra with base set $\mathcal{\ottnt{X}}$,
with top, bottom, join, meet and negation corresponding to $\ottnt{X}$, $ \emptyset $,
union, intersection and complementation respectively. On morphisms, define
\[  F_\oc  \, \ottnt{f} \, \ottsym{(}   X   \ottsym{)} =  F_\wn  \, \ottnt{f} \, \ottsym{(}   X   \ottsym{)} =  \{  \ottnt{f} \, \ottsym{(}  \ottnt{x}  \ottsym{)}  \mid   \ottnt{x}  \in  X   \} . \]

It is easy to check that these functors respect the dualities in that
$ \dualize{ \ottsym{(}   F_\oc  \, \ottnt{P}  \ottsym{)} }  \simeq  F_\wn  \,  \star{ \ottnt{P} } $
and
$ \star{  \bangfunctor{ \ottnt{A} }  }   \simeq   \whynotfunctor{  \dualize{ \ottnt{A} }  } $.
To prove $ F_\oc   \dashv   \bangfunctor{ \ottsym{-} } $, it suffices to show a
bijection of homomorphism sets $\Hom{ F_\oc  \, \ottnt{P},\ottnt{A}} \cong \Hom{\ottnt{P}, \bangfunctor{ \ottnt{A} } }$.
Suppose $\ottnt{f}  \ottsym{:}   F_\oc  \, \ottnt{P}  \rightarrow  \ottnt{A}$ in $ \cat{L} $. Then define $ \ottnt{f} ^\sharp   \ottsym{:}  \ottnt{P}  \rightarrow   \bangfunctor{ \ottnt{A} } $ 
by
\[
     \ottnt{f} ^\sharp  \, \ottsym{(}  \ottnt{x}  \ottsym{)} = \ottnt{f} \, \ottsym{(}    \{   \ottnt{z}  \in  X   \mid  \ottnt{z}  \le  \ottnt{x}  \}    \ottsym{)}
\]
This morphism is in fact order-preserving.
Next, for $\ottnt{g}  \ottsym{:}  \ottnt{P}  \rightarrow   \bangfunctor{ \ottnt{A} } $ define $ \ottnt{g} ^\flat   \ottsym{:}   F_\oc  \, \ottnt{P}  \rightarrow  \ottnt{A}$ as follows:
\[
     \ottnt{g} ^\flat  \,  \ottsym{(}  X  \ottsym{)}  =  \bigvee_{  \ottnt{x}  \in  X  }  \ottnt{g} \, \ottsym{(}  \ottnt{x}  \ottsym{)} 
\]
Again $ \ottnt{g} ^\flat $ is a lattice homomorphism.
It remains to check that  $ \ottsym{(}   \ottnt{f} ^\sharp   \ottsym{)} ^\flat =\ottnt{f}$ and $ \ottsym{(}   \ottnt{g} ^\flat   \ottsym{)} ^\sharp =\ottnt{g}$. 

From these definitions, the unit $  \epsilon  _{ \ottnt{A} }   \ottsym{:}   F_\oc  \,  \bangfunctor{ \ottnt{A} }   \rightarrow  \ottnt{A}$
and counit $  \eta  _{ \ottnt{P} }   \ottsym{:}  \ottnt{P}  \rightarrow   \bangfunctor{  F_\oc  \, \ottnt{P} } $ of the adjunction are 
\[
      \epsilon  _{ \ottnt{A} }  \, \ottsym{(}   X   \ottsym{)} =    \textrm{id}  _{  \bangfunctor{ \ottnt{A} }  }  ^\flat  \, \ottsym{(}   X   \ottsym{)} =  \bigvee_{  \ottnt{x}  \in  X  }    \textrm{id}  _{  \bangfunctor{ \ottnt{A} }  }  \, \ottsym{(}  \ottnt{x}  \ottsym{)}  =  \bigvee  X 
    \qquad\qquad
      \eta  _{ \ottnt{P} }  \, \ottsym{(}  \ottnt{x}  \ottsym{)} =  \ottsym{(}    \textrm{id}  _{  F_\oc  \, \ottnt{P} }   \ottsym{)} ^\sharp  \, \ottsym{(}  \ottnt{x}  \ottsym{)} =  \{  \ottnt{z}  \mid  \ottnt{z}  \le  \ottnt{x}  \} 
\]
To show the adjunction is monoidal, it suffices to prove $ \epsilon $ and $ \eta $
are monoidal natural transformations.

The proof of the comonoidal adjunction $ \whynotfunctor{ \ottsym{-} }   \dashv   F_\wn $ is similar.

\section{Related work}
\label{sec:related}
Girard \cite{Girard87} first introduced linear logic to mix
the constructivity of intuitionistic propositional logic
with the duality of classical logic.
Partly because of this constructivity, there
has been great interest in the semantics of linear logic
in both the classical and intuitionistic fragments.
Consequently, there exist several
categorical frameworks for its semantic models.

One influential framework is Benton \etal's \textit{linear
  category}~\cite{BentonBHP93term}, consisting of
a symmetric monoidal closed category with products and
a linear exponential comonad $ \oc $. 
Schalk \cite{Schalk04} adapted linear categories to the classical case
by requiring that the symmetric monoidal closed category be *-autonomous. 
The coproduct $ \msum $ and coexponential $ \wn $ are then induced from the duality. 

Cockett and Seely \cite{CockettSeely97}, 
seeking to study $ \mprod $ and $ \msum $ as independent structures
unobscured by duality, introduced linearly distributive categories, which make up
the linear category in the $ \textsc{LPC} $ model. The authors extended this motivation
to the exponentials by modeling $ \oc $ and $ \wn $ as linear functors \cite{Cockett99}, 
meaning that $ \wn $ is \emph{not} derived from $ \oc $ and $ \dualize{ \ottsym{(}  \ottsym{-}  \ottsym{)} } $.
The $ \textsc{LPC} $ model reflects that work by allowing $ \oc $ and $ \wn $ 
to have different adjoint decompositions.

Other variations of classical linear logic, notably Girard's
Logic of Unity~\cite{Girard93}, distinguish linear propositions from persistent ones.
In the intuitionistic case, Benton \cite{Benton94mixed}
developed the linear/non-linear logic and categorical model described in \Section{othermodels}.
Barber used this model as the semantics for a term calculus
called $ \textsc{DILL} $~\cite{Barber96}.
A Lafont category \cite{Lafont88} is a canonical instance of an $ \textsc{LNL} $ model
where $\oc  \ottnt{A}$ is the free commutative comonoid generated by $\ottnt{A}$.
This construction automatically admits an adjunction between 
between a linear category $ \cat{L} $ and the 
category of commutative comonoids over $ \cat{L} $.
However, the $ \textsc{LNL} $ and $ \textsc{LPC} $ models have an advantage over Lafont categories
by allowing a much greater range of interpretations for the exponential. 
Lafont's construction excludes traditional models of linear logic like
coherence spaces and the category $ \Rel $.

\medskip 
\noindent\textbf{Acknowledgments.}\quad
The authors thank Beno\^{i}t Valiron, 
Paul-Andr{\'e} Melli{\`e}s, and Marco Gaboardi
for their insights during discussions about this work.
This material is based in part upon work supported by the National Science Foundation Graduate 
Research Fellowship under Grant No. DGE-1321851 and by NSF Grant No. CCF-1421193.


\bibliographystyle{eptcs}
\bibliography{linear}

\begin{thebibliography}{10}
\providecommand{\bibitemdeclare}[2]{}
\providecommand{\surnamestart}{}
\providecommand{\surnameend}{}
\providecommand{\urlprefix}{Available at }
\providecommand{\url}[1]{\texttt{#1}}
\providecommand{\href}[2]{\texttt{#2}}
\providecommand{\urlalt}[2]{\href{#1}{#2}}
\providecommand{\doi}[1]{doi:\urlalt{http://dx.doi.org/#1}{#1}}
\providecommand{\bibinfo}[2]{#2}

\bibitemdeclare{techreport}{Barber96}
\bibitem{Barber96}
\bibinfo{author}{Andrew \surnamestart Barber\surnameend}
  (\bibinfo{year}{1996}): \emph{\bibinfo{title}{Dual intuitionistic linear
  logic}}.
\newblock \bibinfo{type}{Technical Report} \bibinfo{number}{ECS-LFCS-96-347}.

\bibitemdeclare{inproceedings}{BentonBHP93term}
\bibitem{BentonBHP93term}
\bibinfo{author}{Nick \surnamestart Benton\surnameend}, \bibinfo{author}{G.~M.
  \surnamestart Bierman\surnameend} \& \bibinfo{author}{J.~Martin E.~Hyland
  \surnamestart andValeria~de Paiva\surnameend} (\bibinfo{year}{1993}):
  \emph{\bibinfo{title}{A term calculus for intuitionistic linear logic}}.
\newblock In: {\sl \bibinfo{booktitle}{Proceedings of the International
  Conference on TypedLambda Calculi and Applications}},
  \bibinfo{publisher}{Springer-Verlag LNCS 664}, pp. \bibinfo{pages}{75--90},
  \doi{10.1007/BFb0037099}.

\bibitemdeclare{inproceedings}{Benton94mixed}
\bibitem{Benton94mixed}
\bibinfo{author}{P.~N. \surnamestart Benton\surnameend} (\bibinfo{year}{1995}):
  \emph{\bibinfo{title}{A mixed linear and non-linear logic: proofs, terms and
  models}}.
\newblock In: {\sl \bibinfo{booktitle}{Proceedings of Computer Science Logic,
  Kazimierz, Poland.}}, \bibinfo{publisher}{Springer-Verlag}, pp.
  \bibinfo{pages}{121--135}, \doi{10.1007/BFb0022251}.

\bibitemdeclare{article}{Birkhoff37}
\bibitem{Birkhoff37}
\bibinfo{author}{Garrett \surnamestart Birkhoff\surnameend}
  (\bibinfo{year}{1937}): \emph{\bibinfo{title}{Rings of sets}}.
\newblock {\sl \bibinfo{journal}{Duke Mathematical Journal}}
  \bibinfo{volume}{3}(\bibinfo{number}{3}), pp. \bibinfo{pages}{443--454},
  \doi{10.1215/S0012-7094-37-00334-X}.

\bibitemdeclare{article}{CockettSeely97}
\bibitem{CockettSeely97}
\bibinfo{author}{J.R.B. \surnamestart Cockett\surnameend} \&
  \bibinfo{author}{R.A.G. \surnamestart Seely\surnameend}
  (\bibinfo{year}{1997}): \emph{\bibinfo{title}{Weakly distributive
  categories}}.
\newblock {\sl \bibinfo{journal}{Journal of Pure and Applied Algebra}}
  \bibinfo{volume}{114}(\bibinfo{number}{2}), pp. \bibinfo{pages}{133 -- 173},
  \doi{10.1016/0022-4049(95)00160-3}.

\bibitemdeclare{article}{Cockett99}
\bibitem{Cockett99}
\bibinfo{author}{J.R.B. \surnamestart Cockett\surnameend} \&
  \bibinfo{author}{R.A.G. \surnamestart Seely\surnameend}
  (\bibinfo{year}{1999}): \emph{\bibinfo{title}{Linearly distributive
  functors}}.
\newblock {\sl \bibinfo{journal}{Journal of Pure and Applied Algebra}}
  \bibinfo{volume}{143}(\bibinfo{number}{1–3}), pp. \bibinfo{pages}{155 --
  203}, \doi{10.1016/S0022-4049(98)00110-8}.

\bibitemdeclare{article}{Ehrhard05}
\bibitem{Ehrhard05}
\bibinfo{author}{Thomas \surnamestart Ehrhard\surnameend}
  (\bibinfo{year}{2005}): \emph{\bibinfo{title}{Finiteness spaces}}.
\newblock {\sl \bibinfo{journal}{Mathematical Structures in Computer Science}}
  \bibinfo{volume}{15}(\bibinfo{number}{4}), pp. \bibinfo{pages}{615--646},
  \doi{10.1017/S0960129504004645}.

\bibitemdeclare{article}{Fox76}
\bibitem{Fox76}
\bibinfo{author}{Thomas \surnamestart Fox\surnameend} (\bibinfo{year}{1976}):
  \emph{\bibinfo{title}{Coalgebras and cartesian categories}}.
\newblock {\sl \bibinfo{journal}{Communications in Algebra}}
  \bibinfo{volume}{4}(\bibinfo{number}{7}), pp. \bibinfo{pages}{665--667},
  \doi{10.1080/00927877608822127}.

\bibitemdeclare{article}{Girard87}
\bibitem{Girard87}
\bibinfo{author}{Jean-Yves \surnamestart Girard\surnameend}
  (\bibinfo{year}{1987}): \emph{\bibinfo{title}{Linear logic}}.
\newblock {\sl \bibinfo{journal}{Theoretical Computer Science}}
  \bibinfo{volume}{50}(\bibinfo{number}{1}), pp. \bibinfo{pages}{1--101},
  \doi{10.1016/0304-3975(87)90045-4}.

\bibitemdeclare{article}{Girard93}
\bibitem{Girard93}
\bibinfo{author}{Jean-Yves \surnamestart Girard\surnameend}
  (\bibinfo{year}{1993}): \emph{\bibinfo{title}{On the unity of logic}}.
\newblock {\sl \bibinfo{journal}{Annals of Pure and Applied Logic}}
  \bibinfo{volume}{59}(\bibinfo{number}{3}), pp. \bibinfo{pages}{201--217},
  \doi{10.1016/0168-0072(93)90093-S}.

\bibitemdeclare{inproceedings}{kelly1974doctrinal}
\bibitem{kelly1974doctrinal}
\bibinfo{author}{G~Max \surnamestart Kelly\surnameend} (\bibinfo{year}{1974}):
  \emph{\bibinfo{title}{Doctrinal adjunction}}.
\newblock In: {\sl \bibinfo{booktitle}{Category Seminar}},
  \bibinfo{organization}{Springer}, pp. \bibinfo{pages}{257--280},
  \doi{10.1007/BFb0063105}.

\bibitemdeclare{article}{Lafont88}
\bibitem{Lafont88}
\bibinfo{author}{Yves \surnamestart Lafont\surnameend} (\bibinfo{year}{1988}):
  \emph{\bibinfo{title}{The linear abstract machine}}.
\newblock {\sl \bibinfo{journal}{Theoretical Computer Science}}
  \bibinfo{volume}{59}, pp. \bibinfo{pages}{157--180},
  \doi{10.1016/0304-3975(88)90100-4}.
\newblock \bibinfo{note}{Corrections in vol. 62, pp. 327--328}.

\bibitemdeclare{article}{Mellies03}
\bibitem{Mellies03}
\bibinfo{author}{Paul-Andr{\'e} \surnamestart Melli{\`e}s\surnameend}
  (\bibinfo{year}{2003}): \emph{\bibinfo{title}{Categorical models of linear
  logic revisited}}.

\bibitemdeclare{inproceedings}{Mel09}
\bibitem{Mel09}
\bibinfo{author}{Paul-Andr{\'e} \surnamestart Melli{\`e}s\surnameend}
  (\bibinfo{year}{2009}): \emph{\bibinfo{title}{Categorical semantics of linear
  logic}}.
\newblock In: {\sl \bibinfo{booktitle}{Interactive Models of Computation and
  Program Behaviour, Panoramas et Synthèses 27, Société Mathématique de
  France 1–196}}.

\bibitemdeclare{techreport}{PaykinZdancewic14tech}
\bibitem{PaykinZdancewic14tech}
\bibinfo{author}{Jennifer \surnamestart Paykin\surnameend} \&
  \bibinfo{author}{Steve \surnamestart Zdancewic\surnameend}
  (\bibinfo{year}{2014}): \emph{\bibinfo{title}{A linear/producer/consumer
  model of classical linear logic}}.
\newblock \bibinfo{type}{Technical Report} \bibinfo{number}{MS-CIS-14-03},
  \bibinfo{institution}{University of Pennsylvania}.

\bibitemdeclare{article}{Pratt94}
\bibitem{Pratt94}
\bibinfo{author}{V.R. \surnamestart Pratt\surnameend} (\bibinfo{year}{1994}):
  \emph{\bibinfo{title}{Chu spaces: complementarity and uncertainty in rational
  mechanics}}.
\newblock {\sl \bibinfo{journal}{Course notes, TEMPUS summer school,
  Budapest}}.

\bibitemdeclare{misc}{Schalk04}
\bibitem{Schalk04}
\bibinfo{author}{Andrea \surnamestart Schalk\surnameend}
  (\bibinfo{year}{2004}): \emph{\bibinfo{title}{Whats is a categorical model of
  linear logic}}.

\bibitemdeclare{book}{Stanley11}
\bibitem{Stanley11}
\bibinfo{author}{Richard~P \surnamestart Stanley\surnameend}
  (\bibinfo{year}{2011}): \emph{\bibinfo{title}{Enumerative combinatorics}}.
\newblock \bibinfo{publisher}{Cambridge University Press},
  \doi{10.1017/CBO9781139058520}.

\bibitemdeclare{incollection}{ValironZdancewic14}
\bibitem{ValironZdancewic14}
\bibinfo{author}{Benoît \surnamestart Valiron\surnameend} \&
  \bibinfo{author}{Steve \surnamestart Zdancewic\surnameend}
  (\bibinfo{year}{2014}): \emph{\bibinfo{title}{Finite Vector Spaces as Model
  of Simply-Typed Lambda-Calculi}}.
\newblock In: {\sl \bibinfo{booktitle}{Theoretical Aspects of Computing 2014}},
  {\sl \bibinfo{series}{Lecture Notes in Computer Science}}
  \bibinfo{volume}{8687}, \bibinfo{publisher}{Springer International
  Publishing}, pp. \bibinfo{pages}{442--459},
  \doi{10.1007/978-3-319-10882-7\_26}.

\end{thebibliography}

\end{document}